\documentclass[pra,aps,twocolumn,showpacs,nofootinbib]{revtex4}

\usepackage{graphicx,color,epstopdf}
\usepackage{amsmath,amsfonts,enumerate,amsthm}
\usepackage{hyperref}

\hypersetup{
   colorlinks=true,         
   linkcolor=blue,          
   citecolor=blue           
}

\newtheorem{theorem}{Theorem}
\newtheorem{lemma}{Observation}
\newtheorem{definition}{Definition}
\newtheorem{definitionApp}{Definition A\!\!}
\newtheorem{lem}{Lemma}

\newtheorem{corol}{Corollary}

\def\>{\rangle}
\def\<{\langle}
\def\H{ {\cal H} }

\def\Sc{ {\cal S} }
\def\B{ {\cal B} }

\newcommand{\bra}[1]{\langle {#1} |}
\newcommand{\ket}[1]{| {#1} \rangle}
\newcommand{\ketbra}[2]{\ensuremath{\left|#1\right\rangle\!\!\left\langle#2\right|}}
\newcommand{\braket}[2]{\ensuremath{\left\langle#1\right|\left.#2\right\rangle}}

\newcommand{\tr}[2]{\mathrm{Tr}_{#1}\left( #2 \right)}
\newcommand{\iden}{\mathbb{I}}

\renewcommand{\v}[1]{\ensuremath{\boldsymbol #1}}

\begin{document}

\title{Classical noise and the structure of minimal uncertainty states}
\author{Kamil Korzekwa}
\affiliation{Department of Physics, Imperial College London, London SW7 2AZ, United Kingdom}
\author{Matteo Lostaglio}
\affiliation{Department of Physics, Imperial College London, London SW7 2AZ, United Kingdom}

\begin{abstract}
Which quantum states minimise the unavoidable uncertainty arising from the non-commutativity of two observables? The immediate answer to such a question is: it depends. Due to the plethora of uncertainty measures there are many answers. Here, instead of restricting our study to a particular measure, we present plausible axioms for the set $\mathcal{F}$ of bona-fide information-theoretic uncertainty functions. Then, we discuss the existence of states minimising uncertainty with respect to all members of $\mathcal{F}$, i.e., universal minimum uncertainty states (MUS). We prove that such states do not exist within the full state space and study the effect of classical noise on the structure of minimum uncertainty states. We present an explicit example of a qubit universal MUS that arises when purity is constrained by introducing a threshold amount of noise. For higher dimensional systems we derive several no-go results limiting the existence of noisy universal MUS. However, we conjecture that universality may emerge in an approximate sense. We conclude by discussing connections with thermodynamics, and highlight the privileged role that non-equilibrium free energy $F_2$ plays close to equilibrium.
\end{abstract}

\pacs{03.65.-w, 03.65.Aa, 03.67.-a, 05.70.-a, 89.70.Cf}

\maketitle

\section{Introduction}

Since the advent of quantum mechanics uncertainty relations have played a major role in uncovering the mysteries of the quantum realm. Originally introduced by Heisenberg as an error-disturbance relation in his famous thought experiment~\cite{heisenberg1949physical}, they have been reformulated and refined in various ways over the last 90 years. Starting as a statement about the outcome statistics of independent measurements of position and momentum~\cite{kennard1927quantenmechanik}, they were quickly extended to generic observables~\cite{robertson1929uncertainty}. Then, almost half a century after the original formulation, a substantial change in the paradigm came with David Deutsch's proposal of using the Shannon entropy of the outcomes statistics to study uncertainty relations~\cite{deutsch1983uncertainty}. This lead to state-independent lower bounds on uncertainty~\cite{maassen1988generalized}, but also initiated the ongoing intimate relationship between uncertainty relations and the field of information theory. Since then many works used multiple entropic measures to quantify uncertainty (see Ref.~\cite{wehner2010entropic,coles2015entropic} and references therein). Although they shed light on new aspects of uncertainty relations, most of these results depend on the particular choice of the uncertainty measure. A recent ``universal'' approach~\cite{partovi2011majorization,friedland2013universal,puchala2013majorization,rudnicki2014strong} tries to go beyond this limitation, by making statements that are independent of the particular measure of uncertainty being used.

The authors of Ref.~\cite{friedland2013universal} proposed a minimal requirement that all valid uncertainty functions should satisfy: the act of forgetting information about a random variable cannot decrease its uncertainty. This approach restricts uncertainty functions to the family of Schur-concave functions. Here, we also analyse the effect of two additional requirements: the additivity of uncertainty for independent random variables and continuity. Adding these further restrictions one after the other gives two more frameworks for studying uncertainty. These allow us to restrict the set of all uncertainty measures to the one-parameter family of R{\'e}nyi entropies $H_\alpha$ with parameter $\alpha\in\mathbb{R}$ or $\alpha\in\mathbb{R}_+$.

Within these three frameworks we study the structure of quantum states that minimise uncertainty, the so-called minimum uncertainty states (MUS). In the spirit of the universal approach we ask: what quantum states - if any - simultaneously minimise all possible uncertainty measures? The evidence we bring in this paper suggests that, excluding the qubit case, such universal MUS do not exist. Nevertheless, we clarify how the structure of MUS simplifies with the introduction of classical uniform noise and suggest that an approximate notion of universal MUS may emerge. Finally, we link our results on the measures of uncertainty with the problem of quantifying the departure of a system from thermodynamic equilibrium. In particular, we point out how a so-far neglected measure of non-equilibrium, the free energy functional $F_2$ (defined by the R{\'e}nyi divergence of order two), plays a crucial role in near-equilibrium thermodynamic transformations.

\section{General families of uncertainty measures}
\label{sec:generalFamily}

Intuitively, an uncertainty measure $u$ is a function that assigns a real positive number to every probability distribution $\v{p}$, reflecting the ``spread'' of  $\v{p}$. However, there is no unique way of measuring the uncertainty of a probability distribution; quite the contrary, there exists a plethora of different information-theoretic functions~\cite{renyi1961measures}. This is linked to the fact that there are different ways of assessing uncertainty and making bets, depending on the rules of the probabilistic game being played. For example, making a bet on a single event is very different from making bets on many, independent and identically distributed ones. In the former case one would look at a single-shot entropy, whereas in the latter one may choose the Shannon entropy. Also, depending on the stake, one may want to follow a very risk-adverse strategy (and, e.g., look at the Hartley entropy $H_0$) or, on the contrary, be risk-prone (and, e.g., look at the min-entropy $H_\infty$). This is reflected by different choices of the relevant uncertainty functions, as each of them captures a different aspect of the ``spread'' of $\v{p}$. However one can ask: what is the set of all possible uncertainty functions? 

The basic idea is that all uncertainty functions must satisfy some elementary requirements; e.g., all of them should assign zero uncertainty to the sharp probability distribution \mbox{$\v{p}=(1,0,\dots 0)$}. In what follows we will describe and motivate conditions defining general families of bona-fide uncertainty functions. In this paper we will call a probability distribution $\v{p}$ universally less uncertain than $\v{q}$ -- according to some chosen set of uncertainty measures $\mathcal{F}$ -- if $u(\v{p}) \leq u(\v{q})$ for all $u\in\mathcal{F}$.

\subsection{Minimal requirement of Schur-concavity}

Recently, a general condition has been proposed that a function $u$ should satisfy in order to measure uncertainty~\cite{friedland2013universal}. It is given by
\begin{equation}
\label{eq:axiom}
   u(\v{p})\leq u(\lambda\v{p}+(1-\lambda)\Pi\v{p}) \quad \textrm{for~} \lambda\in[0,1],
\end{equation}
where $\Pi$ is any permutation of the probability vector. In other words, a random relabelling of a probability distribution cannot decrease the uncertainty. Notice that since permutations are reversible, this immediately implies that any $u$ must be a function of the probability vector only and not of the way we label events, i.e., $u(\v{p})=u(\Pi\v{p})$. This is in accordance with a much older concept, introduced by Deutsch~\cite{deutsch1983uncertainty}, that an information-theoretic measure of uncertainty for a given observable should not depend on its eigenvalues.

As Birkchoff's theorem states that the convex hull of permutation matrices is given by the set of bistochastic matrices $\{\Lambda\}$~\cite{bhatia1997matrix}, the above axiom is equivalent to
\begin{equation}
\label{eq:maj_order} 
u(\v{p})\leq u(\Lambda\v{p})\mathrm{~for~all~bistochastic~ matrices~}\Lambda.
\end{equation}
Notice~\cite{bhatia1997matrix} that $\v{q}= \Lambda \v{p}$ if and only if $\v{p}$ majorises $\v{q}$, $\v{p} \succ \v{q}$ (we recall the definition of majorisation in Appendix~A). Therefore, functions satisfying Eq.~\eqref{eq:maj_order}, the Shannon entropy being the best known example, are Schur-concave. We shall denote this set by $\Sc$. Hence, the condition given by Eq.~\eqref{eq:axiom} specifies that a minimal requirement for $u$ to be a bona-fide uncertainty function is to be Schur-concave. In Ref.~\cite{friedland2013universal} no further properties are imposed, i.e, it is assumed that actually any $u\in\Sc$ can be considered as a meaningful uncertainty function. Thus, within this approach a probability distribution $\v{p}$ is universally less uncertain than $\v{q}$ if and only if $\v{p} \succ \v{q}$.

\subsection{Enforcing context-independence restricts to R{\'e}nyi entropies}
\label{sec:restricting}
 
In this paper we note that not all Schur-concave functions may be appropriate uncertainty measures, as some of them possess potentially undesired properties. In particular, one can show that there exist probability distributions $\v{p}$, $\v{q}$ and $\v{r}$ such that
\begin{enumerate}
\item \mbox{$\exists u\in\Sc: u(\v{p})>u(\v{q})$},
\item \mbox{$\forall u\in\Sc: u(\v{p}\otimes \v{r})\leq u(\v{q}\otimes\v{r})$}.
\end{enumerate}
This is a simple consequence of the phenomenon of catalysis~\cite{jonathan1999entanglement}. Therefore, allowing any Schur-concave function to measure uncertainty leads to the existence of a measure $u$ according to which $\v{p}$ is more uncertain than $\v{q}$, but \mbox{$\v{p}\otimes\v{r}$} is less uncertain than \mbox{$\v{q}\otimes\v{r}$}.\footnote{As a particular example one can take: $\v{p}=(0.5,0.25,0.25,0)$, $\v{q}=(0.4,0.4,0.1,0.1)$, $\v{r}=(0.6,0.4)$ and $u$ to be the sum of two smallest elements of a probability vector \cite{jonathan1999entanglement}.} As a result, the uncertainty functions are allowed to be context-dependent, i.e., an independent random variable $\v{r}$ can change our assessment of which of two probability distributions, $\v{p}$ or $\v{q}$, is more uncertain. 
Here we will be interested in uncertainty functions that are context-independent, in the sense that independent events do not affect the uncertainty ordering between probability distributions. 

In order to remove context-dependence we propose a single and natural additional assumption restricting the set of allowed measures of uncertainty. We require that all bona-fide measures of uncertainty should not only be Schur-concave, but also additive:
\begin{equation*}
u(\v{p}\otimes\v{q})=u(\v{p})+u(\v{q}).
\end{equation*}
The above condition reflects the extensiveness of uncertainty for independent events, a standard assumption for information and uncertainty measures~\cite{renyi1961measures}. Thus, we define the general family of uncertainty functions by the set of additive Schur-concave functions and denote it by $\mathcal{U}$. It is straightforward to check that by getting rid of non-additive functions the problem of context-dependence is solved. Indeed, due to additivity, for any $u \in \mathcal{U}$ we have
\begin{equation*}
u(\v{p})> u(\v{q}) \Leftrightarrow u(\v{p}\otimes\v{r}) > u(\v{q}\otimes\v{r}) \quad \forall \v{r}.
\end{equation*}
As before we can ask when one random variable is universally less uncertain than another. The answer is that if $\v{p}$ is not simply a permutation of $\v{q}$ this is the case if and only if
\begin{equation}\label{eq:renyi_order}
H_\alpha(\v{p})< H_\alpha(\v{q}) \quad \forall \alpha\in\mathbb{R},
\end{equation}
where $H_\alpha$ are the well-known R{\'e}nyi entropies, first introduced by R{\'e}nyi in Ref.~\cite{renyi1961measures} as general measures of information (and, hence, of uncertainty). We recall their definition in Appendix~A. The fact that Eq.~\eqref{eq:renyi_order} implies the same inequality for all $u \in \mathcal{U}$ is non-trivial and it is a consequence of the results of Refs.~\cite{klimesh2007inequalities,turgut2007catalytic} that show the equivalence between Eq.~\eqref{eq:renyi_order} and the trumping relation~$\succ_T$. A probability distribution $\v{p}$ is said to ``trump'' $\v{q}$ if there exists a context in which $\v{p}$ majorises $\v{q}$ (see Appendix~A for the definition). From such equivalence the result follows immediately. Therefore, choosing $\mathcal{U}$ as the set of uncertainty functions, we can alternatively say that $\v{p}$ is universally less uncertain than $\v{q}$ if and only if $\v{p}$ trumps $\v{q}$, $\v{p} \succ_T \v{q}$.

\subsection{Enforcing decidability restricts to R{\'e}nyi entropies of positive order}
\label{sec:restricting2}

We will now show that allowing R\'{e}nyi entropies of non-positive order $\alpha\leq 0$ to measure uncertainty leaves us with an important problem of undecidability; more precisely, arbitrarily small changes in the probability of events can switch our assessment of which between two probability distributions is more uncertain. Note that this is a physically significant issue, as any physical experiment allows us to determine the probability of events only up to an arbitrarily small, but non-zero error. Hence, using such non-continuous uncertainty functions may lead to a situation in which we need to change our assessment of which of two probability distributions is more uncertain according to an unobservable event. 

To illustrate this problem, let us consider the example of two distributions $\v{p} \otimes \v{r}$ and $\v{q} \otimes \v{r}$ (both with full support) and fix $\alpha < 0$. Then if \mbox{$H_\alpha (\v{p}) > H_\alpha (\v{q})$} we have that $\v{p}\otimes\v{r}$ is \emph{more} uncertain than $\v{q}\otimes\v{r}$ according to the chosen measure $H_\alpha$. However, as we can only know the probabilities of events up to an arbitrarily good approximation, the probability distribution $\v{r}$ on the right hand side may actually be $\v{r}^\epsilon$ with \mbox{$||\v{r} - \v{r}^\epsilon||_1\leq \epsilon$} and some arbitrarily small $\epsilon$ (here $||\cdot||_1$ denotes the $\ell_1$ norm). Then by choosing $\v{r}=(1,0)$ and $\v{r}^\epsilon = (\epsilon, 1-\epsilon)$ we get that $\v{p}\otimes\v{r}$ is \emph{less} uncertain than $\v{q}\otimes\v{r}^\epsilon$, according to $H_\alpha$, for any non-zero $\epsilon$, whereas it is more uncertain if $\epsilon =0$ exactly. Hence, our assessment of which probability distribution is more uncertain is reversed by an undecidable fact (i.e., if $\epsilon$ is exactly zero or not). 
 
To overcome the problem of undecidability one can simply require the continuity of uncertainty functions: given any $\v{p}$, for all $\delta$ there should exist $\epsilon$ such that
\begin{equation}
||\v{p}-\v{p}^\epsilon||_1 \leq \epsilon \Longrightarrow |u(\v{p}) - u(\v{p}^\epsilon)| \leq \delta.
\end{equation}
It is then clear that given $\v{p}$ and $\v{q}$, with $u(\v{p})>u(\v{q})$, also all elements of an $\epsilon$-ball around $\v{p}$ are more uncertain than all elements of an $\epsilon$-ball around $\v{q}$, for $\epsilon >0$ small enough. Note that R\'{e}nyi entropies of order $\alpha \leq 0$ and the Burges entropy are not continuous for distributions without full support. Hence, if we decide to exclude measures affected by this problem, we further restrict the set of uncertainty functions to Schur-concave, continuous and additive functions, denoted by $\mathcal{U}_+$. As before, a probability distribution $\v{p}$ is universally less uncertain than $\v{q}$ if and only if $H_\alpha(\v{p})< H_\alpha(\v{q})$ for all $\alpha>0$.

Let us now summarise the main message of this Section. We have defined three families of bona-fide uncertainty functions (see Fig.~\ref{fig:families}) by means of three natural axioms:\footnote{It may be also worth exploring the set of Schur-concave and continuous functions. In fact, if any of the three presented conditions may be dropped and still give a physically reasonable framework, this seems to be Axiom~\ref{axiom2}.}
\begin{enumerate}
\item Non-increasing under random relabelling, Eq.~\eqref{eq:axiom}, as introduced in Ref.~\cite{friedland2013universal}.  \label{axiom1}
\item Additivity for independent random variables.\label{axiom2}
\item Continuity.\label{axiom3}
\end{enumerate}
Within $\mathcal{S}$, $\v{p}$ is universally more uncertain than $\v{q}$ if and only if \mbox{$p \succ q$}; within $\mathcal{U}$, if and only if \mbox{$H_\alpha(\v{p})< H_\alpha(\v{q})$} for all \mbox{$\alpha \in [-\infty, \infty]$}; and within $\mathcal{U}_+$, if and only \mbox{$H_\alpha(\v{p})< H_\alpha(\v{q})$} for all $\alpha >0$.

\begin{figure}[t!]
\includegraphics[width=\columnwidth]{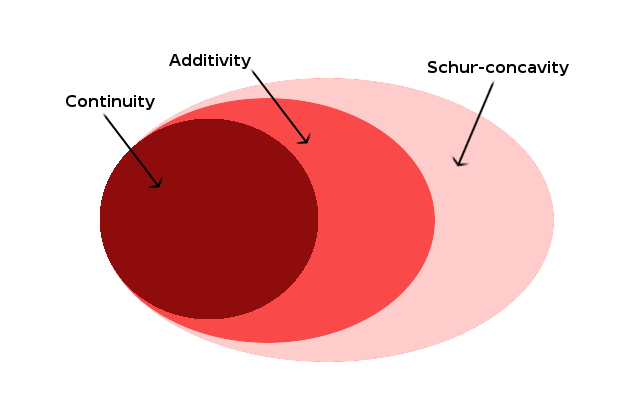}
\caption{\label{fig:families} 
Three possible families of uncertainty functions: Schur-concave functions $\mathcal{S}$ (satisfying Axiom~\ref{axiom1}), R\'{e}nyi entropies of any real order $\mathcal{U}$ (satisfying Axioms~\ref{axiom1}-\ref{axiom2}) and finally R\'{e}nyi entropies of positive order $\mathcal{U}_+$ (satisfying Axioms~\ref{axiom1}-\ref{axiom3}).
}
\end{figure}

\section{Minimum uncertainty states}
\label{sec:minimum}

Having identified the set of conditions that characterise when one probability distribution is universally more uncertain than another, we now have a general framework to study the uncertainty relations. We can investigate the unavoidable uncertainty of the outcome statistics for two non-commuting observables independently of the chosen uncertainty measure. In particular, we will be interested in answering the following question: are there quantum states that simultaneously minimise uncertainty with respect to all uncertainty measures?

In fact, any uncertainty function $u$ defines an uncertainty relation for given observables $A$ and $B$:
\begin{equation}
\label{eq:ubound}
u(\v{p}^A(\rho) \otimes \v{p}^B(\rho))\geq c^u_{AB}\quad \forall \rho\in\B'_d.
\end{equation}
Here $\B'_d$ is some subset of the set of $d$-dimensional quantum states $\B_d$ (often $\B'_d=\B_d$),  $\v{p}^A(\rho)$ and $\v{p}^B(\rho)$ denote the probability distributions over the outcomes of measurements $A$ and $B$ on state $\rho$ and \mbox{$c^u_{AB}>0$} is a constant that does not depend on $\rho \in \B'_d$. As an example, consider the well-known Maassen-Uffink uncertainty relation~\cite{maassen1988generalized}, where $u$ is chosen to be the Shannon entropy $H_1$ and $\B'_d = \B_d$,
\begin{eqnarray*}
	&H_1(\v{p}^A(\rho))+H_1(\v{p}^B(\rho))\geq c^{H_1}_{AB},&\\
	&c^{H_1}_{AB}=-2\ln\left(\max_{ij} |\<a_i|b_j\>|\right),&
\end{eqnarray*}
with $\ket{a_i}$ and $\ket{b_j}$ denoting eigenstates of $A$ and $B$, respectively.

\subsection{Universal minimum uncertainty states}
\label{sec:minimum_1}

States minimising the left-hand side of Eq.~\eqref{eq:ubound} for some choice of $u$ and $\B'_d$ will be called minimum uncertainty states (MUS). These have been found in the case of $u=H_1$ (the Shannon entropy)~\cite{sanches1998optimal,ghirardi2003optimal} and $u=H_2$ (collision entropy)~\cite{bosyk2012collision} for $\B'_d=\B_2$ (qubit systems). However, if we restrict the study of uncertainty relations to a particular uncertainty function, then anything we can say about the structure of MUS will, in general, not hold for a different measure. In this work, having argued for the general sets of bona-fide uncertainty measures {$\mathcal{S}$, $\mathcal{U}$ and $\mathcal{U}_+$}, we can introduce the notion of \emph{universal} minimum uncertainty state: a state that minimises all $u \in \mathcal{F}$ simultaneously, with $\mathcal{F}$ being one of the three sets of uncertainty measures introduced in the previous section. 
\begin{definition}[Universal MUS]
\label{universalmus}
A universal minimum uncertainty state within a subset $\B_d'\subseteq\B_d$ is a state \mbox{$\rho\in\B_d'$} that is universally less uncertain than any other element in $\B_d'$ (modulo permutations of the outcomes). More precisely,
\begin{equation*}
u(\v{p}_A(\rho)\otimes\v{p}_B(\rho))< u(\v{p}_A(\sigma)\otimes\v{p}_B(\sigma)),
\end{equation*}
for all $u \in \mathcal{F}$ and for all $\sigma\in\B_d'$ such that \mbox{$\pi[\v{p}_A(\rho)\otimes\v{p}_B(\rho)]\neq \v{p}_A(\sigma)\otimes\v{p}_B(\sigma)$}, with $\pi$ being an arbitrary permutation.\footnote{Note that there are no probability distributions, other than those linked by permutations, that have the same uncertainty with respect to all $U \in \mathcal{F}$.}  
\end{definition}
\noindent The existence of such special states is conceptually very intriguing. Does quantum mechanics permit their existence for some natural choice of $\B_d'$? An obvious choice that we will consider is the full state space $\B_d'=\B_d$. However, we will also focus on another physically motivated subset of states $\B_d'$ arising while studying the uncertainty relations in the presence of classical noise (where $\B_d'\subset \B_d$ is chosen to be a subset with a given level of mixedness~\cite{luo2005heisenberg,park2005improvement,rudnicki2014strong,korzekwa2014quantum}).

Investigating the existence of universal MUS, or even some approximate version of them -- briefly discussed later in Sec.~\ref{sec:approximate} -- can also be practically relevant. This is because universal MUS simultaneously minimise all possible uncertainty measures over a considered set of states, and different $u$ are operationally relevant in different situations. As uncertainty relations have a range of applications in cryptography and quantum information~\cite{wehner2010entropic}, we conjecture that universal MUS may be useful when we want to perform a protocol, but we do not know in advance what the rules of the probabilistic game are. 	

In what follows we first provide a general no-go theorem forbidding the existence of universal MUS within the full state space $\B_d$. This shows that the best (least uncertain) state always depends on the details of the probabilistic game being played; no ultimate top element exists. However, in many physically relevant scenarios the only available quantum states are mixed. Hence, in the next section, we will explore features emerging from the interplay between non-commutativity and noise.

\subsection{No-go theorem for pure universal MUS}

In the case of two commuting observables (or more generally observables sharing an eigenstate) the existence of a universal MUS is trivial: any common eigenstate has a sharp distribution with respect to both measurements. However, the problem is non-trivial for observables that do not share an eigenstate. In fact, the following result shows that in this case no top element exists within the full unconstrained state space $\B_d$ for all three choices of $\mathcal{F}$:
\begin{theorem}
\label{theorem:noPureUMUS}
Given observables $A$ and $B$ acting on \mbox{$d$-dimensional} Hilbert space $\H_d$ and not sharing any common eigenstate, no universal MUS within the full state space $\B_d$ exists.
\end{theorem}
\noindent The proof consists of two parts. First, we prove that if there exists a universal MUS $\rho\in\B_d$ then it must be pure. Next, we find all states $\{\ket{\psi^{\infty}_i}\}$ that minimise $H_\infty$ and show that there exist pure states that have smaller $H_\alpha$ than any of the $\{\ket{\psi^{\infty}_i}\}$ for some $\alpha>0$. Therefore we conclude that no state can simultaneously minimise all uncertainty measures $u\in\mathcal{F}$ over the full state space $\B_d$. The technical details of the proof can be found in Appendix~B.

Although universal MUS do not exist within the full state space, they still may appear in many physical scenarios where some degree of noise is unavoidable. Noise may be present due to inevitable imperfections in the experimental apparatus, or because the system under scrutiny is entangled with some other degrees of freedom we do not have access to. Hence, one is left to wonder whether the no-go result we derived is robust to noise and, more generally, what is the effect of noise on the structure of minimal uncertainty states. We will discuss this in the next section, starting from some conceptual remarks about assessing uncertainty in the presence of classical uniform noise.

\section{Noise and uncertainty}

\subsection{The role of noise and $H_2$ in the classical case}

To build up intuition as to why the introduction of noise can make a difference in assessing uncertainty, it is useful to start with a simple yet suggestive example. Whereas the importance of $H_0$, $H_1$ and $H_\infty$ R{\'e}nyi entropies has been previously stressed~\cite{gour2015resource}, here we emphasize the special role played by the collision entropy $H_2$ in the presence of noise. Consider two probability distributions
\begin{equation*}
\v{p}=(0.77,0.10,0.10,0.03), \quad \v{q}=(0.63,0.35,0.01,0.01),
\end{equation*}
and two sources $P$ and $Q$ that produce messages by drawing from a four-element alphabet according to these probability distributions. One can immediately check that \mbox{$H_1(\v{p}) > H_1(\v{q})$}, which means that the messages produced by $Q$ will have a higher compression rate than those produced by $P$. However, now assume that there is an additional noise channel that affects the messages produced by sources $P$ and $Q$, so that the effective probability distributions become
\begin{equation}
\label{eq:noisyprobability}
\v{p}^\epsilon = \epsilon \v{\eta} + (1-\epsilon) \v{p},
\end{equation}
and similarly for $\v{q}^{\epsilon}$, where $\v{\eta} =(0.25,0.25,0.25,0.25)$ is a uniform distribution. It is then easy to verify that, for $\epsilon \geq 0.05$, we have \mbox{$H_1(\v{p}^\epsilon) < H_1(\v{q}^\epsilon)$}. According to the Shannon entropy, $\v{p}$ is more uncertain than $\v{q}$ (and hence more difficult to compress), but the situation is reversed once enough noise is introduced. This shows that the noise more strongly affects the information content of the message produced by $Q$, as measured by $H_1$.

This discussion leads us to the following question: under what conditions the information content encoded in a source $P$ (as measured by a generic $H_\alpha$) is more strongly affected by uniform noise than the information encoded in $Q$? The answer is provided by the following result:
\begin{lemma}
\label{lemma:ordering_by_2}
Given two probability distributions $\v{p}$ and $\v{q}$, with $H_2(\v{p})\neq H_2(\v{q})$, the following statements are equivalent for any given $\alpha\in(-\infty,\infty)$:
\begin{enumerate}
\item There exists $\epsilon_\alpha \in [0,1)$ such that
\begin{equation*}
H_{\alpha}(\v{p}^\epsilon) < H_{\alpha}(\v{q}^\epsilon) \quad \forall \epsilon\geq\epsilon_\alpha.
\end{equation*}
\item $H_2(\v{p}) < H_2(\v{q})$.
\end{enumerate}
\end{lemma}
\noindent The simple proof is based on the fact that, for every $\alpha$, $H_{\alpha}(\v{p}^\epsilon)$ has a maximum at $\epsilon=1$ (as $\v{p}^\epsilon$ then corresponds to a uniform distribution), so that the dominating term in the Taylor expansion around this maximum is quadratic in the probabilities $\{p_i\}$. Then for finite $\alpha$ the ordering between $H_{\alpha}(\v{p}^\epsilon)$ and $H_{\alpha}(\v{q}^\epsilon)$ for $\epsilon$ close enough to 1 depends solely on the ordering of the $\alpha=2$ R{\'e}nyi entropy (see Appendix~C for the details). Note, however, that the ordering for the limiting cases of $\alpha=\pm \infty$ can never be changed by introducing noise (see Appendix~D). 

Observation~\ref{lemma:ordering_by_2} shows that noise can indeed play a crucial role in uncertainty relations: it induces an order within the set of probability distributions, with $H_2$ playing a leading role. However, the problem is more complicated than one might initially expect. Recall that the crucial question unanswered by Observation~\ref{lemma:ordering_by_2} is whether a ``finite'' amount of noise is sufficient to induce an ordering between all Renyi entropies; in other words, if there exists an $\tilde{\epsilon}<1$ \emph{independent} of $\alpha$ such that \mbox{$H_\alpha(\v{p}^{\tilde{\epsilon}})<H_\alpha(\v{q}^{\tilde{\epsilon}})$} for all $\alpha$. In fact, given two generic probability distributions, the condition \mbox{$H_\alpha(\v{p}) < H_\alpha(\v{q})$} for $\alpha=2$ and $\alpha=\pm\infty$ is necessary, but \emph{not} sufficient, to induce ordering between R{\'e}nyi entropies for all $\alpha$. A counterexample is given by $\v{p}=(0.37,0.32,0.24,0.07)$ and $\v{q}=(0.36,0.35,0.19,0.10)$. A direct calculation shows that $H_2(\v{p})<H_2(\v{q})$, but for any amount of noise $\epsilon$ we have $H_\alpha(\v{p}) > H_\alpha(\v{q})$ for $\alpha=4/(1-\epsilon)$. This implies that for $\alpha \rightarrow \infty$ the required amount of noise must go to $1$, so there is no single $\tilde{\epsilon}<1$ that ensures the relation \mbox{$H_\alpha(\v{p}^{\tilde{\epsilon}})<H_\alpha(\v{q}^{\tilde{\epsilon}})$} is satisfied for all $\alpha$.

\subsection{The role of noise and $H_2$ in the quantum case}

Given the discussion and results above, it is natural to define a quantum analogue of Eq.~\eqref{eq:noisyprobability} by the set $\epsilon$-noisy states $\B_d^\epsilon$ that can be written in the form
\begin{equation*}
\B_d^\epsilon := \left\{\rho^\epsilon:\quad \rho^\epsilon = \epsilon \iden/d+(1-\epsilon)\rho, \; \; \rho \in \B_d \right\},
\end{equation*}
for a generic state $\rho$ and a fixed $\epsilon\in[0,1]$. However, by the same reasoning as in Theorem~\ref{theorem:noPureUMUS}, i.e., using the strong concavity and additivity of the Shannon entropy, one can show that among all states in $\B_d^\epsilon$ only the ones for which $\rho$ is pure can be universal MUS. This means that considerations concerning universal MUS can be restricted to the set of pseudo-pure states,\footnote{Note that instead of considering a projective measurement described by projectors $\{\ketbra{a_i}{a_i}\}$ on pseudo-pure states, one can equivalently consider a \textit{noisy} positive operator valued measure (POVM) with POVM elements \mbox{$\{\epsilon \iden/d+(1-\epsilon)\ketbra{a_i}{a_i}\}$}.} first introduced in the field of NMR spectroscopy~\cite{cory1997ensemble}: 
\begin{definition}[Pseudo-pure states]
\label{noisystates}
A state belongs to the subset of $\epsilon$-pseudo-pure states if it can be written in the form
\begin{equation}\label{eq:noisy_state}
\rho^\epsilon_\psi=\epsilon \iden/d+(1-\epsilon)\ket{\psi}\bra{\psi}, \quad \epsilon\in [0,1].
\end{equation} 
\end{definition}

We now provide a modified version of Observation~\ref{lemma:ordering_by_2} suited for probability distributions arising from the measurement of two non-commuting observables in the presence of noise. Let us define \small
\begin{equation*}
\Delta H_{\alpha}:= H_{\alpha}(\v{p}^A(\rho^\epsilon) \otimes \v{p}^B(\rho^\epsilon))- H_{\alpha}(\v{p}^A(\sigma^\epsilon) \otimes \v{p}^B(\sigma^\epsilon)).
\end{equation*}\normalsize
We then have the following:
\begin{lemma}
\label{observation3}
Let $\rho$ and $\sigma$ denote any two quantum states and $A$ and $B$ any two observables. If 
\mbox{$e^{-H_2(\v{p^A}(\rho))} + e^{-H_2(\v{p^B}(\rho))} \neq e^{-H_2(\v{p^A}(\sigma))} + e^{-H_2(\v{p^B}(\sigma)}$}, the following two conditions are equivalent for any given $\alpha$:
\begin{enumerate}
\item \label{condition1obs} There exists $ \epsilon_{\alpha}$: 
\begin{equation*} 
\Delta H_{\alpha} < 0, \quad \forall \epsilon \geq \epsilon_\alpha,
\end{equation*}
\item \label{condition2obs} \small{$e^{-H_2(\v{p^A}(\rho))} + e^{-H_2(\v{p^B}(\rho))} > e^{-H_2(\v{p^A}(\sigma))} + e^{-H_2(\v{p^B}(\sigma))}$}.
\end{enumerate}
\end{lemma}
\noindent The proof can be found in Appendix~E. For any given measure of uncertainty $H_\alpha$, this observation shows that the knowledge of $H_2$ is sufficient to answer the following question: ``which of two states has outcome statistics of two non-commuting measurements more uncertain in the presence of large enough uniform noise?''.

Nevertheless, similarly to the classical case, we have no guarantee that a finite amount of noise will generate an ordering between all R\'enyi entropies. In the next section we will explore this question.

\section{Existence of noisy universal MUS}

\subsection{General results}

Let us start by presenting three general results concerning noisy universal MUS that are valid for all three choices of $\mathcal{F}$. One will give us an explicit candidate for such state; the other two prevent the existence of universal MUS in a broad set of situations.

\subsubsection{A candidate universal MUS}

First, we provide a technical lemma that may be of interest independently from the question of finding noisy universal MUS:
\begin{lem}
\label{lemma:Hinf}
Given observables $A$ and $B$, the \mbox{$\epsilon$-pseudo-pure} state minimising \mbox{$H_\infty(\v{p}^A\otimes\v{p}^B)$} is given by 
\begin{equation}
\label{eq:rho_Hinf}
\rho^\epsilon_{\psi_\infty} = \epsilon \iden/d + (1-\epsilon)\ketbra{\psi_\infty}{\psi_\infty},
\end{equation}
with $\ket{\psi_\infty} \propto \ket{a_i}+e^{-i \phi}\ket{b_j}$, $\ket{a_i}$ and $\ket{b_j}$ being the eigenstates of $A$ and $B$ that maximise $|\<a_i|b_j\>|$ and \mbox{$\phi=\arg \<a_i|b_j\>$}.
\end{lem}
\noindent The proof is presented in Appendix~F and follows a route similar to the one used in proving Theorem~\ref{theorem:noPureUMUS}. Lemma~\ref{lemma:Hinf} immediately singles out a candidate for universal MUS by providing its explicit form:

\begin{corol}
	\label{corol:hinf}
Given observables $A$ and $B$ that do not share a common eigenstate, if there exists a universal MUS in $\B_d^\epsilon$, then it must be $\epsilon$-pseudo-pure with pure state $\ket{\psi_\infty}$.
\end{corol}

\subsubsection{No-go results}

Let us now present the first no-go result concerning mutually unbiased observables:
\begin{theorem}[No-go for mutually unbiased observables]\label{lemma:mub}
Given observables $A$ and $B$ that are mutually unbiased, no universal MUS exists within $\B_d^\epsilon$, for any \mbox{$\epsilon \in [0,1)$}.
\end{theorem}
\begin{proof}
	Note that the Shannon entropy \mbox{$H_1(\v{p}^A(\rho)\otimes\v{p}^B(\rho))$} is minimised among $\epsilon$-noisy states by $\epsilon$-pseudo-pure state given in Eq~\eqref{eq:noisy_state} with $\ket{\psi}$ being an eigenstate of either $A$ or $B$. This is because such states saturate the tight bound \mbox{$\ln d+S(\rho)$} found in Ref.~\cite{korzekwa2014quantum} for mixed states in the case of mutually unbiased observables. Then, by direct calculation one can check that for an $\epsilon$-pseudo-pure state defined in Lemma~\ref{lemma:Hinf}, the Shannon entropy is higher. Hence such a state cannot be a universal MUS and so, from Corollary~\ref{corol:hinf}, no universal MUS exists. 
\end{proof}

For higher dimensional systems $d\geq 3$ we can provide a general no-go result severely limiting the existence of universal MUS in the presence of noise for the choices $\mathcal{F}=\mathcal{S}$ and $\mathcal{F}=\mathcal{U}$. Specifically, we have:
\begin{theorem}
		Consider two observables $A$ and $B$ with eigenstates $\{\ket{a_i}\}$ and $\{\ket{b_j}\}$ such that \mbox{$V_{ij}=\langle a_i | b_j \rangle\neq 0$}. Then, if the dimension of the system is $d\geq 3$, the introduction of noise does not lead to the emergence of universal MUS for $\mathcal{F}=\mathcal{S}$ and $\mathcal{F}=\mathcal{U}$.
	\end{theorem}
	\begin{proof}
		From Corollary~\ref{corol:hinf} we know that a noisy universal MUS must be of the form specified by Eq.~\eqref{eq:rho_Hinf}, so that
		\begin{equation}
		\label{eq:hminus}
		H_{-\infty}(\v{p}^A(\rho^\epsilon_{\psi_\infty}) \otimes \v{p}^B (\rho^\epsilon_{\psi_\infty})) > 2 \ln \epsilon/d,
		\end{equation}
		the inequality being strict because $V_{ij}\neq 0$. Now consider the state \mbox{$\rho^\epsilon_\zeta = \epsilon \iden/d+(1-\epsilon)\ket{\xi}\bra{\xi}$}, where \mbox{$\ket{\xi} = a_1 \ket{a_1} + a_2 \ket{a_2}$} and $a_1$, $a_2$ are chosen such that $\ket{\xi}$ is orthogonal to $\ket{b_1}$. Then computing the left hand side of Eq. \eqref{eq:hminus} for $\sigma^\epsilon$ gives exactly $2 \ln \epsilon/d$. Hence, the state minimising $H_{-\infty}$ does not coincide with $\rho^\epsilon_{\psi_\infty}$. From Corollary~\ref{corol:hinf} it implies that no universal MUS exists for $\mathcal{F}=\mathcal{S}$ and $\mathcal{F}=\mathcal{U}$.
	\end{proof}

The general results above say nothing about two distinct qubit observables when they are not mutually unbiased. In fact, as often happens, qubits are special and thus will be investigated in the next section. We will then conclude by suggesting a possible approximate notion of universality for higher dimensions and the choice \mbox{$\mathcal{F}=\mathcal{U}_+$}.

\subsection{Universal MUS for qubit systems}
\label{sec:qubit}

As all considered uncertainty measures depend only on the eigenstates of the observables and not on their eigenvalues, without loss of generality we can choose qubit observables $A=\v{a}\cdot\v{\sigma}$ and $B=\v{b}\cdot\v{\sigma}$, where $\v{\sigma}$ denotes the vector of Pauli operators, while $\v{a}$ and $\v{b}$ are the Bloch vectors. Let us also denote the angle between the two Bloch vectors by $\gamma$, so that \mbox{$\v{a}\cdot\v{b}=\cos\gamma$}.

\begin{figure}[t!]
\includegraphics[width=\columnwidth]{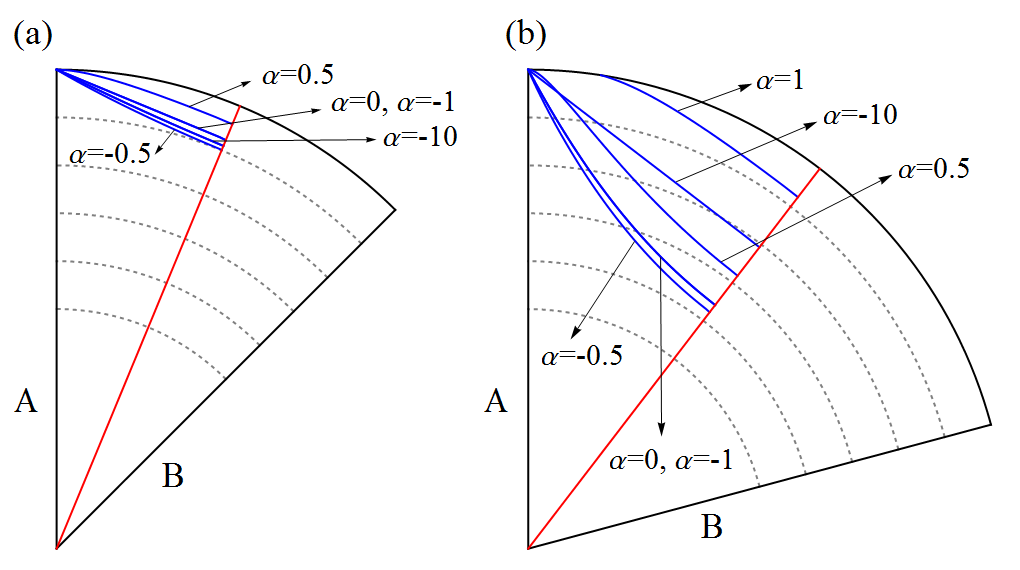}
\caption{\label{fig:qubit_MUS} 
The blue curves described by arrows show the position of MUS with respect \mbox{$\alpha$-R{\'e}nyi} entropies of different orders, as a function of the purity (distance from the origin) for qubit observables $A$ and $B$. $A$ and $B$ are separated by angle $\gamma=\pi/4$ [panel (a)] and $\gamma=5\pi/12$ [panel (b)], and the plot presents the Bloch sphere in the first quadrant of the plane spanned by eigenvectors of $A$ and $B$. The bisecting line (in red) denotes the set of states approached by all MUS and, ultimately, containing the universal MUS. For simplicity we only plot MUS for $\theta\in[0,\gamma/2]$, as the case $\theta\in[\gamma/2,\gamma]$ is symmetric. Note that MUS with respect to $\alpha\geq 1$ [panel (a)] and $\alpha\geq 2$ [panel (b)] lie on the bisecting line for all purities. Dashed grey lines correspond to states with fixed purity (length of the Bloch vector equal to $1-n/10$ for \mbox{$n\in\{1,\dots 5\}$}).}
\end{figure} 

The detailed analysis and calculations can be found in Appendix~H and here we will only state the main results. First of all, for the choice $\mathcal{F}=\mathcal{S}$ (the framework of majorisation uncertainty relations) there is no universal MUS in the presence of noise, i.e., no amount of noise $\epsilon<1$ can lead to the emergence of such state. However, for the two other choices of the family of uncertainty functions ($\mathcal{F}=\mathcal{U}$ and $\mathcal{F}=\mathcal{U}_+$) universal MUS may emerge after introducing a threshold amount of noise. Specifically, in Appendix~G we prove that for $\gamma=\pi/4$ the amount of noise $\epsilon=1/2$ leads to the emergence of universal MUS.

In panel (a) of Fig. \ref{fig:qubit_MUS} we illustrate this emergence of a universal MUS with the introduction of noise. We plot the states that for a given purity minimise different \mbox{$\alpha$-R{\'e}nyi} entropies. As can be seen in the Figure, above a threshold level of noise all the R{\'e}nyi entropies are minimised by a state described by a Bloch vector lying on the bisection of the angle $\gamma$. Note that, according to a numerical investigation, the level of noise $\epsilon=1/2$ used in the proof is actually much larger than required. In panel (b) of Fig. \ref{fig:qubit_MUS} we similarly plot the position of MUS for different $\alpha$, but in the case of qubit observables separated by $\gamma=5\pi/12$. Notice that now the amount of noise required for universal MUS to appear is larger and, from Theorem~\ref{lemma:mub}, one can expect that it grows with $\gamma$ up to the point when $\epsilon=1$ for $\gamma=\pi/2$ (corresponding to mutually unbiased bases). A numerical investigation supports the conjecture that universal MUS exist for generic qubit observables.

\subsection{Approximate notion of universality for higher dimensions}
\label{sec:approximate}

For dimensions $d\geq3$ the no-go theorems presented so far do not apply to the choice $\mathcal{F}=\mathcal{U}_+$ (corresponding to uncertainty measures that are both continuous and context-independent) when the considered observables are not mutually unbiased. As in higher dimensions the analytical verification of the emergence of $\epsilon$-noisy universal MUS becomes extremely complicated, we numerically verify whether universal MUS could emerge with the introduction of noise. The following lemma provides an additional necessary condition for such emergence that can be checked numerically by only investigating pure states:
\begin{lem}
	\label{theorem:H2}
	A necessary condition for a universal MUS to emerge with the introduction of noise is that the expression
	\begin{equation}
	\label{eq:num_cond}
	e^{-H_{2}(\v{p}^A(\rho))} + e^{-H_2(\v{p}^B(\rho))} 
	\end{equation}
	is maximised among pure states by $\rho=\ketbra{\psi_\infty}{\psi_\infty}$ (for the definition of $\ket{\psi_\infty}$ see Lemma~\ref{lemma:Hinf}).
\end{lem}
\begin{proof}
	By contradiction, suppose that there exists a pure state $\phi \equiv \ketbra{\phi}{\phi}$ such that
	\begin{equation*}
	e^{-H_{2}(\v{p}^A(\psi_\infty))} + e^{-H_2(\v{p}^B(\psi_\infty))} < e^{-H_{2}(\v{p}^A(\phi))} + e^{-H_2(\v{p}^B(\phi))},
	\end{equation*}
	where $\psi_\infty \equiv \ketbra{\psi_\infty}{\psi_\infty}$. Let \mbox{$\rho^\epsilon_{\psi_\infty} = \epsilon \iden/d + (1-\epsilon)\psi_\infty$} and \mbox{$\rho^\epsilon_{\phi} = \epsilon \iden/d + (1-\epsilon) \phi$} be the corresponding \mbox{$\epsilon$-pseudo-pure} states and fix \mbox{$\alpha=1$}. Then by Observation~\ref{observation3} there exists $\tilde{\epsilon}$ such that for all $\epsilon\geq\tilde{\epsilon}$, one has $\Delta H_1 > 0$. This implies that, for any of the three choices of $\mathcal{F}$, there is no $\epsilon$ close enough to 1 such that $\rho^\epsilon_{\psi_\infty}$ is a universal MUS in $B^\epsilon_d$. However, by Corollary~\ref{corol:hinf}, this implies that a universal MUS does not exist within $B^\epsilon_d$ for any amount $\epsilon$ of noise introduced.
\end{proof}

Given two observables $A$ and $B$ we can now use the above result to numerically verify whether the universal MUS emerges with the introduction of noise. This can be done in the following way. First, we need to find $\ket{\psi_\infty}$. Then, using a numerical optimization procedure we search for a state $\ket{\psi_{\mathrm{opt}}}$ that minimises Eq.~\eqref{eq:num_cond}. Finally we can compare $\ket{\psi_\infty}$ with $\ket{\psi_{\mathrm{opt}}}$ and if these states differ we can conclude that no universal MUS exists for $A$ and $B$, even for the choice $\mathcal{F}=\mathcal{U}_+$.

We numerically investigate $d\in\{3,4,5\}$, each time generating 1000 pairs of observables $(A,B)$ whose eigenvectors are connected by a unitary, randomly chosen according to the Haar measure. Our analysis shows that $\ket{\psi_\infty}$ does not coincide with $\ket{\psi_{{\rm opt}}}$, showing that in general no universal MUS exists in higher dimension even with the choice $\mathcal{F}=\mathcal{U}_+$.\footnote{Even more strongly, no universal MUS exist whenever $\mathcal{F}$ contains both $H_\infty$ and any $H_\alpha$ for finite $\alpha$, e.g., $\mathcal{F}= \{H_1,H_\infty\}$.} However, we also observe that the two states are very close. More precisely, we found that their average overlap $|\!\braket{\psi_{\mathrm{opt}}}{\psi_\infty}\!|$ is equal to $0.9996$, $0.9904$ and $0.9842$ for dimension $d$ equal to 3, 4 and 5, respectively. From Observation~\ref{observation3} we know that for any given $\alpha \in (0, +\infty)$, $\rho^\epsilon_{\mathrm{opt}} = \epsilon \iden/d + (1-\epsilon)\ketbra{\psi_{\mathrm{opt}}}{\psi_{\mathrm{opt}}}$ has smaller $H_{\alpha}(p^A \otimes p^B)$ than any other given pseudo-pure state if $\epsilon$ is taken to be bigger than some $\epsilon_{\alpha}<1$. So, for any arbitrarily fine sample of $\alpha$'s and pseudo-pure states, there would be some $\tilde{\epsilon}$ small enough such that $\rho^{\tilde{\epsilon}}_{\mathrm{opt}}$ is the best pseudo-pure state. At the same time, the case $\alpha = \infty$ is optimised by the pseudo-pure state $\rho^\epsilon_{\psi_\infty}$ that, as we said above, has pure component with large overlap with $\ket{\psi_{\mathrm{opt}}}$. This leads naturally to the conjecture that noise can lead to the emergence of a universal MUS in an approximate sense. We leave this as an interesting open question for future work.

In Fig. \ref{fig:qutrit_MUS} we present an example of the emergence of such approximate universal MUS for a qutrit system. We choose the observables $A$ and $B$ such that the eigenstates of $B$ are connected to the eigenstates of $A$ by a rotation around $(1,1,1)$ axis by angle $\pi/6$. As can be seen in panel (a) of Fig.~\ref{fig:qutrit_MUS}, without noise the candidate universal MUS $\rho^\epsilon_{\psi_\infty}$ has larger R\'{e}nyi entropies of order $\alpha<1$ than the optimal state. For example, for zero noise the optimal states for $H_{0.1}$ are close to the eigenstate of either $A$ or $B$. However, the introduction of noise $\epsilon=0.25$ results in the approximate equality (discrepancy on the order of $10^{-4}$ at worst) between R\'{e}nyi entropies of the candidate and optimal state for the investigated region of $\alpha\in[0,2]$. Numerical investigations also show that this approximate equality holds for $\alpha>2$.

\begin{figure}[t!]
	\includegraphics[width=\columnwidth]{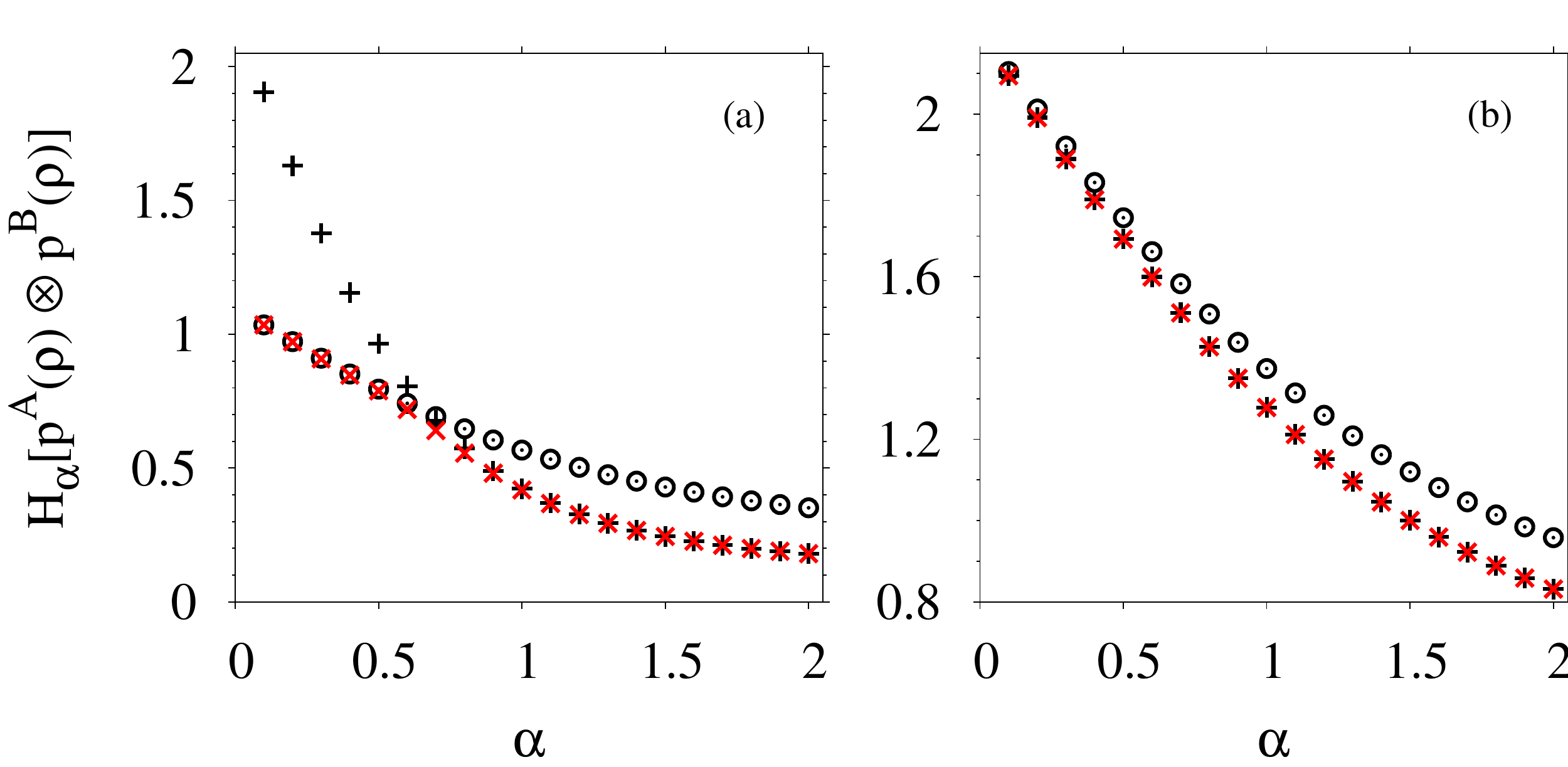}
	\caption{\label{fig:qutrit_MUS} R{\'e}nyi entropies for $\rho^\epsilon_\infty$ (black pluses), the eigenstate of $A$ (black circles) and numerically optimized $\epsilon$-noisy states minimising a given $\alpha$-R{\'e}nyi entropy (red X's), for the qutrit system and observables $A$ and $B$ as described in the main text. (a) Level of noise $\epsilon=0$; (b) Level of noise $\epsilon=0.25$.}
\end{figure}

\section{Thermodynamic considerations}

Although this work is mainly focused on uncertainty relations, we devote this section to point out the links between measures of uncertainty in information theory and measures of the departure from equilibrium in thermodynamics. This allows us to extend our observation on the role of collision entropy to the field of thermodynamics and point out some interesting properties linked to thermalisation.

\subsection{Measuring the departure from equilibrium}

Given a system described by a free Hamiltonian $H_S$ and in contact with a heat bath at inverse temperature \mbox{$\beta=(kT)^{-1}$}, we can introduce the non-equilibrium $\alpha$-free energy functionals~\cite{brandao2013second}
\begin{equation}
\label{eq:freeenergies}
F_{\alpha}(\rho)= - kT \ln Z_S + kT S_{\alpha}(\rho||\gamma_S),
\end{equation}
where $Z_S = \tr{}{e^{-\beta H_S}}$ is the partition function, \mbox{$\gamma_S = e^{-\beta H_S}/ Z_{H_S}$} is the equilibrium Gibbs state of the system and $S_\alpha$ denote $\alpha$-R{\'e}nyi divergences whose definitions we recall in Appendix~A. Notice that for equilibrium states all $F_\alpha$ coincide with the thermodynamic free energy. In fact, $kT S_{\alpha}(\rho||\gamma_S)$ can be interpreted as a non-equilibrium contribution to the free energy. In what follows we will focus only on ``classical'' non-equilibrium states that are diagonal in the energy eigenbasis and can be characterised by the distribution over energy eigenstates $\rho \leftrightarrow \v{p}$. Note that transformations between such states can be described in general by a stochastic matrix $\Lambda$ acting on the probability vector $\v{p}$ describing the state.

Let us first convince the reader that $\alpha$-free energy functionals are not some arbitrary measures of non-equilibrium, but are of fundamental importance and are linked with R{\'e}nyi entropy measures of uncertainty. In fact, a similar reasoning to the one showing that $H_\alpha$ constitute a canonical family of uncertainty functions satisfying two simple and natural axioms can be applied to the quantification of non-equilibrium, with $F_\alpha$ playing the role of $H_\alpha$. In order to introduce the thermodynamic analogue of the axioms \ref{axiom1} and \ref{axiom2} from Sec.~\ref{sec:restricting}, let us denote by $\Lambda_T$ any Gibbs-preserving map, i.e., any stochastic map satisfying $\Lambda_T(\v{\gamma}) = \v{\gamma}$, where $\v{\gamma}$ is the probability vector describing occupation of energy eigenstates in thermal equilibrium. Gibbs-preserving maps are the most general set of transformations between ``classical'' states that can be performed without using work~\cite{faist2015gibbs, lostaglio2015quantum}. In fact, any operation outside this set brings an initially thermal state out of equilibrium, which would allow for building a perpetuum mobile of the second kind by extracting work from a single heat bath, thus violating the second law of thermodynamics. We then require all functions $f$ that quantify departure of the system from thermodynamic equilibrium to satisfy the following two axioms:
\begin{enumerate}
\item \label{eq:freecon1} $f(\Lambda_T(\v{p})) \leq f(\v{p})$.
\item $f(\v{p} \otimes \v{q}) = f(\v{p}) + f(\v{q})$.
\end{enumerate}
The first axiom requires $f$ to be monotonically decreasing under Gibbs-preserving maps. As stated above, Gibbs-preserving maps can be performed at zero work cost. Hence, if we could bring a system farther out from the equilibrium for free (by increasing $f$ using $\Lambda_T$), we could then thermalise it back and extract positive work, thus building a perpetuum mobile. The second axiom requires measures of non-equilibrium to be additive for independent systems. 

Using the reasoning presented in Section~\ref{sec:restricting} and the results of~\cite{ruch1978mixing}, one finds that the first of the above axioms implies that $f$ must respect the ordering induced by a thermodynamic generalisation of the notion of majorisation, called thermo-majorisation~\cite{horodecki2013fundamental} [i.e., if $\v{p}$ thermo-majorizes $\v{q}$ then $f(\v{p})\geq f(\v{q})$]. The second requirement then leads us to a thermodynamic analogue of the notion of trumping, which in~\cite{brandao2013second} was proven to be characterised exactly by the non-equilibrium free energies defined by Eq.~\eqref{eq:freeenergies}. Hence $\{F_\alpha\}$ play the same canonical role in quantifying the departure from equilibrium as the R{\'e}nyi entropies in the case of measuring uncertainty. Finally, notice that by requiring continuity we would restrict to positive order free energies.

The interest in these quantities also relies on several operational interpretations attached to them. The simultaneous decrease of all $\{F_\alpha \}$ is a necessary and sufficient condition for the existence of a thermal operation (defined in~\cite{janzing2000thermodynamic, brandao2011resource}) between two non-equilibrium incoherent states when auxiliary catalysts are allowed \cite{brandao2013second}. The decrease of all positive order free energies was also given an operational interpretation in \cite{brandao2013second}, in terms of catalytic thermal operation where one is allowed to borrow a qubit ancilla that is given back arbitrarily close to its initial state at the end of the transformation. Moreover, the $\alpha=1$ free energy 
\begin{equation*}
F_1(\rho)= \tr{}{\rho H_S} - kT \ln Z_{S},
\end{equation*}
is privileged in various ways: it is a bound for the average work that a system can perform while equilibrating with respect to a bath at temperature $T$~\cite{aberg2013truly}; it governs transformations in the ``thermodynamic limit''~\cite{brandao2011resource}; it was also recently shown to govern transformations where we can access a source of stochastic independence~\cite{lostaglio2015stochastic}.

\subsection{Near-equilibrium thermodynamics}

We will now translate Observation~\ref{lemma:ordering_by_2} from Section~\ref{sec:minimum} into the language of thermodynamics and analyse the consequences for near-equilibrium processes. In order to do this let us exchange the set of $\epsilon$-noisy states $\B^\epsilon_d$ with the set of $\epsilon$-thermal states:
\begin{equation}
\label{eq:epsilon_thermal}
\mathcal{T}^\epsilon_d := \{ \sigma^\epsilon = \epsilon \gamma_S + (1-\epsilon) \rho, \; \; \rho \in \mathcal{B}_d\}.
\end{equation}
Notice that Eq.~\eqref{eq:epsilon_thermal} describes states that are the outcome of an elementary model of thermalisation~\cite{scarani2002thermalising}. We then have the following result:
\begin{lemma}
\label{lemma:thermo_ordering}
Consider two quantum states $\rho$ and $\sigma$ diagonal in the energy eigenbasis (described by distributions $\v{p}$ and $\v{q}$, respectively) with $F_2(\rho)\neq F_2(\sigma)$. Then the following statements are equivalent:
\begin{enumerate}
\item For every $\alpha\in\mathbb{R}$ there exists $\epsilon_\alpha \in[0,1)$ such that
\begin{equation*}
F_{\alpha}(\rho^\epsilon) > F_{\alpha}(\sigma^\epsilon)\quad \forall\epsilon\geq\epsilon_\alpha.
\end{equation*}
\item $F_2(\rho) > F_2(\sigma)$.
\end{enumerate}
\end{lemma}
\noindent The proof is a trivial generalisation of Observation~\ref{lemma:ordering_by_2} and it is hence omitted.

Observation~\ref{lemma:thermo_ordering} provides operational meaning to a so-far neglected thermodynamic quantity: the $\alpha=2$ free energy defined by
\begin{equation*}
F_2(\rho) = - kT \ln Z_{S} + k T \ln \sum_i \frac{p^2_i}{\gamma_i},
\end{equation*}
where $p_i$ are the eigenvalues of $\rho$, \mbox{$\gamma_i=e^{-\beta E_i}/Z_S$} and $\{E_i\}$ is the set of the eigenvalues of $H_S$. Notice that $F_2$ is linked to the thermal average of $(p_i/\gamma_i)^2$. Given any valid measure of non-equilibrium, for $\epsilon$ large enough $F_2$ determines which of two the states, $\rho^\epsilon$ or $\sigma^\epsilon$ [partially thermalised versions of states $\rho$ and $\sigma$, see Eq.~\eqref{eq:epsilon_thermal}], is farther from thermal equilibrium. Hence, $F_2$ provides an ordering between different near-equilibrium states.

As an application of Observation~\ref{lemma:thermo_ordering}, consider a system in state $\rho$ that we want to transform into a target state $\sigma$ by putting it in contact with a heat bath. Let us assume that \mbox{$F_1(\rho) < F_1 (\sigma)$}, so that such transformation is forbidden by the second law of thermodynamics. However, it is still possible that \mbox{$F_2(\rho) > F_2 (\sigma)$}. If this is the case then close enough to equilibrium we have \mbox{$F_1(\rho^\epsilon) > F_1 (\sigma^\epsilon)$}. This implies that we can transform many copies of $\rho^\epsilon$ into many copies of $\sigma^\epsilon$~\cite{brandao2011resource}. It also means that we can extract on average a positive amount of work by transforming $\rho^\epsilon$ into $\sigma^\epsilon$, even though work is required to transform $\rho$ to $\sigma$. Moreover, as mentioned above, for $\epsilon$ large enough we can transform $\rho^\epsilon$ into $\sigma^\epsilon$ by thermal operations using a source of stochastic independence. Finally, note that taking into account the thermalisation interpretation of Eq.~\eqref{eq:epsilon_thermal}, the reversal of free energy $F_1$ ordering between states $\rho$ and $\sigma$ can arise from a thermalisation process. For example, if initial states $\rho(0)$ and $\sigma(0)$, with $F_2(\rho(0)) > F_2(\sigma(0))$, thermalise according to Eq.~\eqref{eq:epsilon_thermal} at the same rate $\epsilon(t)$ [$\epsilon(t)$ monotonically increases with $t$], then for all times $t>t_r$ for some finite $t_r$ we will have \mbox{$F_1(\rho(t)) > F_1 (\sigma(t))$}.

\section{Conclusions and outlook}

Uncertainty relations quantify the impossibility of preparing a quantum state with the statistics of two non-commuting observables being simultaneously sharp. As such, they tell us something fundamental about the quantum world. On the other hand, whenever a given uncertainty relation is based on a specific choice of uncertainty measure, it is biased. Although this choice may be justified by other assumptions, it unavoidably limits the universality of the results obtained. 

The particular example of this problem that we focused on in this paper is the form of minimum uncertainty states (MUS). Already in the simplest case of two qubit observables $A=\v{a}\cdot\v{\sigma}$ and $B=\v{b}\cdot\v{\sigma}$ (using the standard notation introduced in Sec.~\ref{sec:qubit}) one easily finds that MUS are not unique and depend on the chosen measure. Indeed, if as a measure of uncertainty we choose the Shannon entropy of the outcome probabilities, pure MUS may be given by eigenvectors of either $A$ or $B$ \cite{sanches1998optimal,ghirardi2003optimal}; if we choose the min-entropy instead, pure MUS are always described by the Bloch vector lying in the middle between the two closest eigenvectors of $A$ and $B$.

Inspired by the recent ``universal'' approach to uncertainty relations~\cite{partovi2011majorization,friedland2013universal,puchala2013majorization,rudnicki2014strong}, we discussed minimal \emph{desiderata} for the set of uncertainty functions. More precisely, in this work we considered three axioms: monotonicity under random relabelling of events (as in Ref.~\cite{friedland2013universal}), additivity for independent random variables, and continuity. Imposing these requirements one after the other led us to three possible definitions of the set of information-theoretic uncertainty functions.

We then introduced the concept of \emph{universal} minimum uncertainty state. We defined it as quantum state that minimises uncertainty of the measurement outcome distributions of two non-commuting observables with respect to all information-theoretic uncertainty measures belonging to a chosen family. We have found that within the full state space such universal MUS do not exist in finite-dimensional Hilbert spaces for any choice of the family of uncertainty functions. 

This led us to consider the role of noise in uncertainty relations and the structure of minimal uncertainty states, and to identify the crucial role played by R\`enyi entropy of order $\alpha=2$. We have given an explicit example of universal MUS for qubits and argued for their generic existence in such systems. However, for higher dimensional systems we proved strong no-go theorems preventing the existence of universal MUS. Our partial results suggest, however, that an approximate, rather than exact, notion of universality may emerge.

From an operational rather than axiomatic point of view, it may be relevant to consider a family of uncertainty functions $\mathcal{F}$ given by a finite number of R\'enyi entropies with $\alpha \in (-\infty, \infty)$, e.g., \mbox{$\mathcal{F}=\{ H_{1/2}, H_1\}$}. In the presence of noise $\epsilon$, using strict concavity we know that the states minimising $H_{1/2}$ and $H_1$ must be $\epsilon$-pseudo-pure. Then one can numerically verify, e.g., for a qutrit system, that the pure component of $\rho^\epsilon_{\psi_{\rm opt}}$ is very close to the pure components of the states minimising $H_{1/2}$ and $H_1$, for some threshold amount of noise $\epsilon$. This shows that  $\rho^\epsilon_{\psi_{\rm opt}}$ is at least ``almost'' universally optimal with respect to $\mathcal{F}$.

We believe this may be a promising new avenue in the study of uncertainty relations; although each uncertainty measure has its own operational meaning relevant in a specific scenario, for qubit systems (and, in a weaker sense, also for higher dimensions), there exist states that are universally optimal in the presence of a strong enough uniform noise. Our observations may have consequences for quantum cryptography as, e.g., bounds on the knowledge that an eavesdropping party can acquire about the information encoded in non-orthogonal states can be derived from uncertainty relations~\cite{wehner2010entropic}. Hence, we conjecture that our work may be relevant, for example, when the eavesdropper does not have prior information about the number of times a given quantum channel will be used. We leave a deeper analysis of these connections for future work.

\bigskip

\textbf{Acknowledgements:} We would like to thank Antony Milne, Karol \.{Z}yczkowski and Raam Uzdin for helpful comments and discussions. We are also very grateful for the ongoing support provided by David Jennings and Terry Rudolph. Finally, we would like to mention the initial stimulating discussions with Sofia Qvarfort and Henry O'Hagan. This work was supported by EPSRC and in part by COST Action MP1209.

\section*{Appendix A - Mathematical background}

\subsection*{Majorisation}

Given a $d$-dimensional probability vector $\v{p}$, we will denote by $\v{p}^\downarrow$ the same vector but with elements rearranged in decreasing order. We now recall the definition of majorisation~\cite{marshall1979inequalities}
\begin{definitionApp}[Majorisation]
Given two probability distributions $\v{p}$ and $\v{q}$, we say that $\v{p}$ majorises $\v{q}$, and write $\v{p} \succ \v{q}$, if and only if
\begin{equation*}
\sum_{i=1}^n p^\downarrow_i \geq \sum_{i=1}^n q^\downarrow_i, \quad n=1,...,d-1.
\end{equation*}
\end{definitionApp}

\subsection*{R{\'e}nyi entropies}

R{\'e}nyi entropies for $d$-dimensional probability distributions are defined as~\cite{renyi1961measures, mueller2015generalization}
\begin{equation*}\renewcommand{\arraystretch}{1.5}
H_\alpha(\v{p}):=\left\{\begin{array}{cc}
\frac{\mathrm{sgn}(\alpha)}{1-\alpha}\ln\left(\sum_ip_i^\alpha\right)&\mathrm{for~}\alpha\neq 0,\\
\frac{1}{d}\sum_i\ln p_i&\mathrm{for~}\alpha=0.
\end{array}\right.
\end{equation*}
The $\alpha \rightarrow \pm \infty$ and $\alpha\rightarrow 1$ are defined by suitable limits,
\begin{eqnarray*}
H_1(\v{p})&=&-\sum_i p_i \ln p_i,\\
H_\infty(\v{p})&=& - \ln \max_i p_i,\\
H_{-\infty}(\v{p})&=& \ln \min_i p_i.
\end{eqnarray*}
Note that for $\alpha=0$ we defined $H_0$ as the Burg entropy, and not the limit of the R{\'e}nyi entropy for $\alpha\searrow 0$. Notice also the extension to negative $\alpha$.

\subsection*{Trumping}

Trumping is essentially a context-independent version of majorisation. It is sometimes also called catalytic majorisation~\cite{daftuar2001mathematical}: 
\begin{definitionApp}
We say that $\v{p}$ trumps $\v{q}$ and denote it by $\v{p}\succ_T\v{q}$ when 
\begin{equation*}
\exists \v{r}: \v{p}\otimes\v{r}\succ \v{q}\otimes\v{r}.
\end{equation*}
\end{definitionApp}
\noindent The results of~\cite{klimesh2007inequalities,turgut2007catalytic} link trumping and R{\'e}nyi entropies as follows. If $\v{p}\neq\v{q}$ then
\begin{equation*}
\v{p}\succ_T\v{q} \Longleftrightarrow H_\alpha(\v{p}) < H_\alpha(\v{q}), \quad \forall \alpha \in \mathbb{R}.
\end{equation*}
However, from the definition of trumping and the additivity of any Schur-concave additive function $u \in \mathcal{U}$, we have
\begin{equation*}
\v{p}\succ_T \v{q} \Longrightarrow u(\v{p}) < u(\v{q}), \quad \forall u \in \mathcal{U}.
\end{equation*}
This implies
\begin{equation*}
H_\alpha(\v{p}) < H_\alpha(\v{q}), \; \forall \alpha \in \mathbb{R} \Longrightarrow u(\v{p}) < u(\v{q}), \; \forall u \in \mathcal{U},
\end{equation*}
as anticipated in Sec.~\ref{sec:restricting}.

\subsection*{R{\'e}nyi divergences}

Given probability distributions $\v{p}$ and $\v{q}$, the $\alpha$-R{\'e}nyi divergence (or relative entropy) is defined as~\cite{renyi1961measures}:
\begin{equation}
\label{eq:renyidivergence}
S_\alpha(\v{p}||\v{q}) = \frac{\rm{sgn}(\alpha)}{\alpha -1}\ln \sum_i p_i^\alpha q_i^{1-\alpha},
\end{equation}
for $\alpha \neq \{0,1\}$. The values at $\alpha = \{0,1,\pm \infty\}$ are defined through Eq.~\eqref{eq:renyidivergence} by the correspondent limits and read~\cite{brandao2013second}
\begin{eqnarray*}
S_0(\v{p}\|\v{q}) = - \ln \sum_{i | p_i \neq 0} q_i, && S_1(\v{p}\|\v{q}) = \sum_i p_i \ln \frac{p_i}{q_i}, \\
S_{\infty}(\v{p}\|\v{q}) = \ln \max_i \frac{p_i}{q_i}, && S_{-\infty}(\v{p}\|\v{q}) = S_\infty(\v{q}\|\v{p}).
\end{eqnarray*}

\section*{Appendix B - Proof of Theorem~1}

\begin{proof}

Given any mixed state $\rho$, let us decompose it in its own eigenbasis $\{\ket{\psi_k}\}$:
\begin{equation*}
\rho = \sum_k \lambda_k \ketbra{\psi_k}{\psi_k}:=\sum_k \lambda_k \rho_k.
\end{equation*}
Using the notation introduced in Sec.~\ref{sec:minimum_1}, the strict concavity and additivity of Shannon entropy implies that
\begin{eqnarray*}
H(\v{p}^A(\rho)\otimes\v{p}^B(\rho)) &>& \sum_{k,l} \lambda_k\lambda_l H(\v{p}^A(\rho_k)\otimes\v{p}^B(\rho_l))\\
& = & \sum_k \lambda_k H(\v{p}^A(\rho_k)\otimes\v{p}^B(\rho_k))\\
& \geq & \min_k H(\v{p}^A(\rho_k)\otimes\v{p}^B(\rho_k)),
\end{eqnarray*}
so that for every mixed state $\rho$ there exists a pure state $\rho_k$ that is characterised by lower Shannon entropy. This immediately implies that no mixed state can be a universal MUS and hence we can consider only pure states.

We will now find the set of pure states $\{\ket{\psi^m_{\infty}}\}$ that minimise  \mbox{$H_\infty(\v{p}^A(\ket{\psi})\otimes\v{p}^B(\ket{\psi}))$} among all pure states $\ket{\psi}$. Let us remind that \mbox{$H_\infty(\v{p})=-\ln\max_i p_i$}, so that we are actually looking for states that maximise the largest entry of the probability vector \mbox{$\v{p}^A(\ket{\psi})\otimes\v{p}^B(\ket{\psi})$}. Let $V$ be the unitary connecting eigenbases of $A$ and $B$, i.e., $\ket{a_i}=V\ket{b_i}$ for all $i=1,...,d$. Since, by assumption, $A$ and $B$ do not share an eigenstate we have \mbox{$\forall i,j \, \, \,  |V_{ij}|:=|\<b_i|V|b_j\>|<1$}. Let $c$ denote the absolute value of the matrix element of $V$ that has the largest absolute value, i.e.,
\begin{equation}
\label{eq:c}
c:=\max_{i,j} |V_{ij}|:=|V_{i_{M}j_{M}}|<1,
\end{equation}
where $(i_M,j_M)$ denotes the indices corresponding to one of such largest elements of $V$.
Now, let's decompose a general normalized pure state $\ket{\psi}$ into the eigenstates of $A$:
\begin{equation}
\label{eq:psi}
\ket{\psi}=\sum_{k=1}^{d}\alpha_k\ket{a_k}=\sqrt{p}\ket{a_{k_M}}+\sqrt{1-p}\ket{a_{k_M}^{\perp}},
\end{equation}
where \mbox{$\max_k |\alpha_k|=|\alpha_{k_M}|:=\sqrt{p}$}, and we absorbed a phase in the definition of $\ket{a_{k_M}}$. Also,
\begin{equation*}
\ket{a_{k_M}^{\perp}}:=\frac{1}{\sqrt{1-p}}\sum_{k\neq k_M} \alpha_k |a_k\>.
\end{equation*}
Let $p_{\mathrm{max}}(\ket{\psi})$ denote the maximal element of the joint probability distribution \mbox{$\v{p}^A(\ket{\psi})\otimes\v{p}^B(\ket{\psi})$}:
\begin{equation*}
p_{\mathrm{max}}(\ket{\psi}):= \max_{k,l} p^A_k(\ket{\psi}) p^B_l(\ket{\psi})= p\max_{l} p^B_l(\ket{\psi}).
\end{equation*}
We have
\begin{eqnarray*}
p_{\mathrm{max}}(\ket{\psi})&=&\max_l p\left|\sqrt{p}\<b_l|a_{k_M}\>+\sqrt{1-p}\<b_l|a_{k_M}^\perp\>\right|^2 \\
& := & p\left|\sqrt{p}\<b_{l_M}|a_{k_M}\>+\sqrt{1-p}\<b_{l_M}|a_{k_M}^\perp\>\right|^2 \\
&= & p\left|\sqrt{p}|V_{l_M k_M}|+e^{i x } \sqrt{1-p}\sqrt{1-|V_{l_M k_M}|^2}\right|^2 \\
& \leq &p\left(\sqrt{p}|V_{l_M k_M}|+ \sqrt{1-p}\sqrt{1-|V_{l_M k_M}|^2}\right)^2,
\end{eqnarray*}
where $l_M$ is the index $l$ satisfying the first maximisation problem and \mbox{$x:= \arg \<b_{l_M}|a_{k_M}\> - \arg\<b_{l_M}|a_{k_M}^\perp\>$}. The inequality is tight only if $x=0$ and then we have
\begin{equation*}
p_{\mathrm{max}}(\ket{\psi}) \leq \frac{(1+|V_{l_M k_M}|)^2}{4} \leq \frac{(1+c)^2}{4},
\end{equation*}
where the first inequality is attained for \mbox{$p=(1+|V_{l_M k_M}|)/2$} and the second inequality is attained only if $l_M=i_M$ and $k_M=j_M$. One easily finds that the tightness of all of the above inequalities implies that
\begin{equation*}
\ket{b_{i_M}} = e^{i\phi}\left(c \ket{a_{k_M}} + \sqrt{1-c^2} \ket{a^\perp_{k_M}}\right),
\end{equation*}
where \mbox{$\phi=\arg \<a_{k_M}|b_{i_M} \>$}. We can now solve the above equation for $\ket{a^\perp_{k_M}}$ and substitute the result to Eq.~\eqref{eq:psi}. Finally using $k_M=j_M$ and optimal \mbox{$p=(1+c)/2$} one finds that states maximising $p_{\mathrm{max}}(\ket{\psi})$ are of the form
\begin{equation}\label{eq:psi_inf}
\ket{\psi_{\infty}^m}=\frac{\ket{a_{j_M}}+e^{-i\phi}\ket{b_{i_M}}}{\sqrt{2(1+c)}},
\end{equation}
where $m$ enumerates all pairs $(i_M,j_M)$ for which $|V_{ij}|$ attains maximum. It is also worth noting that states of the above form actually saturate the bound found by Landau and Pollak~\cite{landau1961prolate}, for the product of maximum outcome probabilities for non-commuting observables (see Eq. (9) of~\cite{maassen1988generalized}).

Finally, we just need to show that there exists a pure state $\ket{\psi}$ for which \mbox{$H_\alpha(\v{p}^A(\ket{\psi})\otimes\v{p}^B(\ket{\psi}))$} is smaller than for any of the states $\{\ket{\psi_{\infty}^m}\}$ for some $\alpha>0$. This can be proved in the following way. First, define \mbox{$\tilde{H}_0(\v{p})=\ln |\mathrm{supp}~\v{p}|$}, where $|\mathrm{supp}~\v{p}|$ denotes the number of non-zero elements of $\v{p}$. Then, note that the distribution \mbox{$\v{p}^A(\ket{\psi_{\infty}^m})\otimes\v{p}^B(\ket{\psi_{\infty}^m})$} has full support, so that $\tilde{H}_0$ has a value of $\ln d^2$. On the other hand, the probability distribution corresponding to any eigenstate of $A$ or $B$ has at most $d$ non-zero entries, so that $\tilde{H}_0$ for such states is smaller than or equal to $\ln d$. Finally, note that \mbox{$\lim_{\alpha\rightarrow 0}H_\alpha(\v{p})=\tilde{H}_0(\v{p})$} and $H_\alpha$ is continuous in $\alpha>0$, which means that there exists $\alpha>0$ such that $H_\alpha$ is bigger for any of $\{\ket{\psi_{\infty}^m}\}$ than for any of the eigenstates of either $A$ or $B$. 
\end{proof}

\section*{Appendix C - Proof of Observation~1}

A simple calculation shows
\begin{equation*}
H_{\alpha}(\v{p}^\epsilon) - H_{\alpha}(\v{q}^\epsilon) = \frac{\rm{sgn}(\alpha)}{1-\alpha} \ln A,
\end{equation*}
where
\begin{equation*}
A = \frac{\sum_i [r p_i d + 1-r]^\alpha}{\sum_i [r q_i d + 1-r]^\alpha}.
\end{equation*}
and $r:=1-\epsilon$. Hence $H_{\alpha}(\v{p}^\epsilon) < H_{\alpha}(\v{q}^\epsilon)$ is equivalent to $A>1$ for $\alpha < 0$ and $\alpha > 1$, whereas for $\alpha\in(0,1)$ it is equivalent to  $A < 1$. Expanding around $r =0$ one gets
\begin{equation*}
\sum_i [r x_i d + 1-r]^\alpha = d + \frac{\alpha(\alpha -1)}{2}\left(\sum_i x_i^2d^2 -d\right) r^2 + O(r^3),
\end{equation*}
where $\{x_i\}$ denotes the entries of either $\v{p}$ or $\v{q}$. Hence, for any given $\alpha$, we can rewrite $H_{\alpha}(\v{p}^\epsilon) < H_{\alpha}(\v{q}^\epsilon)$ as
\begin{equation*}
\sum_i (p_i^2 - q_i^2) d^2 r^2 + O(r^3) > 0.
\end{equation*}
Then, it is clear that if \mbox{$\sum_i p_i^2 > \sum_i q_i^2$}, i.e., \mbox{$H_2(\v{p}) < H_2(\v{q})$}, the above inequality is satisfied for $r_\alpha>0$ small enough. An analogous proof can be used to show that the statement is valid for $\alpha =0$ and $\alpha=1$. On the other hand, for a given $\alpha$ let \mbox{$H_{\alpha}(\v{p}^\epsilon) < H_{\alpha}(\v{q}^\epsilon)$} for all $r\leq r_\alpha$ with $r_\alpha >0$. Then it must be that \mbox{$\sum_i p_i^2 \geq \sum_i q_i^2$} and so \mbox{$H_2(\v{p}) \leq H_2(\v{q})$}. By assumption however \mbox{$H_2(\v{p}) \neq H_2(\v{q})$}, hence \mbox{$H_2(\v{p}) < H_2(\v{q})$}.

\section*{\texorpdfstring{Appendix D - Ordering of $H_2$ and $H_{\pm\infty}$ is unaffected by noise}{Appendix D - Ordering of $H2$ and $H_{+/-infinity}$ is unaffected by noise}}

First note that $H_2(\v{p})<H_2(\v{q})$ is equivalent to \mbox{$\sum_i p_i^2>\sum_i q_i^2$}. Recall that we denote by $p^\epsilon_i$ the elements of \mbox{$\v{p}^\epsilon = (1-\epsilon)\v{p} + \epsilon \v{\eta}$}. Then, introducing $r:=1-\epsilon$, we have
\begin{equation*}
\sum_i \left(p^\epsilon_i\right)^2=\sum_i \left(\frac{1-r}{d}+r p_i\right)^2=\frac{1-r^2}{d}+r^2\sum_i p_i^2.
\end{equation*}
It is then immediate to see that \mbox{$H_2(\v{p})<H_2(\v{q})
\Leftrightarrow H_2(\v{p}^\epsilon)<H_2(\v{q}^\epsilon)$} for all $r \in (0,1]$, i.e., for every $\epsilon \in [0,1)$. 

Similarly note that $H_{\infty}(\v{p})<H_{\infty}(\v{q})$ is equivalent to $\max_i p_i>\max_i q_i$. We will thus consider
\begin{equation*}
\max_i p^\epsilon_i=\max_i \left(\frac{1-r}{d}+r p_i\right)=\frac{1-r}{d}+r\max_ip_i.
\end{equation*}
Again, it is easy to see that \mbox{$H_\infty(\v{p})< H_\infty(\v{q})$} is equivalent to \mbox{$H_\infty(\v{p}^\epsilon)< H_\infty(\v{q}^\epsilon)$} for all $\epsilon \in [0,1)$. An analogous reasoning works for $H_{-\infty}$. Moreover, from Observation~\ref{lemma:ordering_by_2}, no other $H_\alpha$ has this property.

\section*{Appendix E - Proof of Observation~2}

\begin{proof}
Introducing \mbox{$r:=1-\epsilon$} one can compute
\begin{equation*}
\v{p}^A(\rho^\epsilon) \otimes \v{p}^B(\rho^\epsilon) = \frac{(1-2r)}{d^2} + 2r\v{P}^{AB}(\rho) + r^2 \v{Q}^{AB}(\rho),
\end{equation*}
where 
\begin{eqnarray*}
\v{P}^{AB}(\rho)&=&\frac{\v{p}^A(\rho)  + \v{p}^B(\rho)}{2d},\\
\v{Q}_{ij}^{AB}(\rho)&=&\v{p}_i^A(\rho) \v{p}_j^B(\rho) - \frac{\v{p}_i^A(\rho) + \v{p}_j^B(\rho)}{d} + \frac{1}{d^2}.
\end{eqnarray*}
Notice that the same expression holds for $\sigma^\epsilon$. We now proceed as in Appendix~C. We have \mbox{$\Delta H_\alpha = \rm{sgn}(\alpha)\ln B/(1-\alpha)$} with \small
\begin{equation*}
B = \frac{\sum_{ij} [2r \v{P}_{ij}^{AB}(\rho) d^2 + 1-2r + r^2d^2 Q^{AB}_{ij}(\rho)]^\alpha}{\sum_{ij} [2r \v{P}_{ij}^{AB}(\sigma) d^2 + 1-2r + r^2d^2 Q^{AB}_{ij}(\sigma)]^\alpha}:=\frac{g_\alpha(\rho)}{g_\alpha(\sigma)}.
\end{equation*}\normalsize
Hence $\Delta H_{\alpha}\leq 0$ is equivalent to $B>1$ for $\alpha < 0$ and $\alpha > 1$, whereas for $\alpha\in(0,1)$ it is equivalent to  $B < 1$. Expanding around $r =0$ one gets
\begin{equation*}
g_\alpha(\psi) = d^2 + 2\alpha(\alpha -1)\left(\sum_{ij} \v{P}_{ij}^{AB}(\rho)^2 d^4 -d^2\right)r^2 + O(r^3).
\end{equation*}
Therefore, for any $\alpha$ we can rewrite $\Delta H_\alpha < 0$ as
\begin{equation}
\label{eq:equivalent}
\sum_{ij} (\v{P}_{ij}^{AB}(\rho)^2 - \v{P}_{ij}^{AB}(\sigma)^2)r^ 2 d^4 + O(r^3) > 0.
\end{equation}
Let us fix $\alpha$. If Condition~\ref{condition2obs} holds, then \mbox{$\sum_{ij} \v{P}_{ij}^{AB}(\rho)^2 > \sum_{ij} \v{P}_{ij}^{AB}(\sigma)^2$}. Hence, there exists $r_\alpha$ small enough
such that Eq.~\eqref{eq:equivalent} is satisfied, i.e., Condition~\ref{condition1obs} holds. On the other hand, if for a given $\alpha$ Condition~\ref{condition1obs} holds for all $\epsilon \geq \epsilon_{\alpha}$ (i.e., for all $r \leq r_{\alpha}$) then one must have $\sum_{ij} \v{P}_{ij}^{AB}(\rho)^2 \geq \sum_{ij} \v{P}_{ij}^{AB}(\sigma)^2$, which is equivalent to
\begin{equation}
\nonumber
e^{-H_{2}(\v{p}^A(\rho))} + e^{-H_2(\v{p}^B(\rho))} \geq e^{-H_{2}(\v{p}^A(\sigma))} + e^{-H_2(\v{p}^B(\sigma))}.
\end{equation}
However, by assumption the equality does not hold, so we obtain Condition~\ref{condition2obs}.
\end{proof}

\section*{Appendix F - Proof of Lemma~1}

Note that the following proof uses the notation and results obtained while proving Theorem~\ref{theorem:noPureUMUS} in Appendix~B.
\begin{proof}
We want to show that among $\epsilon$-pseudo pure states $\rho^\epsilon_\psi$ defined in Eq.~\eqref{eq:noisy_state}, the ones with $\ket{\psi}=\ket{\psi^m_\infty}$ for some $m$ (see Eq.~\eqref{eq:psi_inf}) are those minimising the quantity \mbox{$H_\infty(\v{p}^A\otimes\v{p}^B)$}. Instead of minimising this entropic quantity we can equivalently maximise over all pure states $\ket{\psi}$ the maximal element $p_{\max}(\rho^\epsilon_\psi)$ of the probability vector \mbox{$\v{p}^A(\rho^\epsilon_\psi)\otimes\v{p}^B(\rho^\epsilon_\psi)$}. We have
\begin{eqnarray*}
p_{\max}(\rho^\epsilon_\psi)&=&\frac{\epsilon^2}{d^2}+\frac{\epsilon(1-\epsilon)}{d}\max_{k,l} \left(p^A_i(\ket{\psi})+p^B_j(\ket{\psi})\right)\\
&+&(1-\epsilon)^2\max_{k,l}\left(p^A_i(\ket{\psi})p^B_j(\ket{\psi}\right).
\end{eqnarray*}
From the proof of Theorem~\ref{theorem:noPureUMUS} (see Appendix~B) we know that the last term is maximised for $\ket{\psi}=\ket{\psi_\infty^m}$. We will now show that the second term is also maximised for the same state and therefore the whole expression for \mbox{$p_{\max}(\rho^\epsilon_\psi)$} is maximised for this choice of $\ket{\psi}$. 

To shorten the notation let us introduce
\begin{equation*}
s_{\max}(\ket{\psi})=\max_{k,l} \left(p^A_k(\ket{\psi})+ p^B_l(\ket{\psi})\right).
\end{equation*}
Using the same reasoning that lead us in the proof of Theorem~\ref{theorem:noPureUMUS} to the bound on $p_{\max}(\ket{\psi})$, we can obtain a bound on $s_{\max}(\ket{\psi})$. More precisely we have
\begin{equation*}
s_{\mathrm{max}}(\ket{\psi}) \leq p+\left(\sqrt{p}|V_{l_M k_M}|+ \sqrt{1-p}\sqrt{1-|V_{l_M k_M}|^2}\right)^2,
\end{equation*}
where we use the same notation as in Appendix~B. It is straighforward to show that the above expression is maximised for $|V_{l_M k_M}|=c$ and $p=(1+c)/2$. Similarly as in the proof of Theorem~\ref{theorem:noPureUMUS}, this leads to the conclusion that $s_{\mathrm{max}}(\ket{\psi})$ is maximised by the states $\ket{\psi_\infty^m}$.
\end{proof}

\section*{Appendix G - Universal MUS for qubit systems}

\subsection*{Setting the scene}

The Bloch sphere can be parametrized so that \mbox{$\v{a}=(0,0,1)$} and \mbox{$\v{b}=(\sin\gamma,0,\cos\gamma)$}, with $\gamma \in (0,\pi/2)$. Now, according to Lemma~\ref{lemma:Hinf}, if there exists $\epsilon$-noisy universal MUS it will be described by
\begin{equation}
\label{eq:qubit_uMUS}
\rho^\epsilon_{\gamma/2}=\frac{\iden+\v{r}\cdot\v{\sigma}}{2},\quad \v{r}=\pm|r| \left(\sin\frac{\gamma}{2},0,\cos\frac{\gamma}{2}\right),
\end{equation}
with $|r|=(1-\epsilon)$, i.e., its Bloch vector $\v{r}$ will lie in the middle between $\pm\v{a}$ and $\pm\v{b}$, and its length will differ from identity by the amount of noise $\epsilon$. 

We will first prove that there is no universal MUS for qubit systems when we choose $\mathcal{F}=\mathcal{S}$, i.e., within the framework of majorisation uncertainty relations. Let $\rho^\epsilon_0$ be a state described by Bloch vector \mbox{$\v{r}'=(0,0,|r|)$}. Then by direct calculation one can check that for any given $\epsilon$ the distribution \mbox{$\v{p}^A(\rho^\epsilon_0)\otimes \v{p}^B(\rho^\epsilon_0)$} is not majorised by \mbox{$\v{p}^A(\rho^\epsilon_{\gamma/2})\otimes \v{p}^B(\rho^\epsilon_{\gamma/2})$}. 

Now, in order to prove that there exists $0<\epsilon <1$ such that $\rho^\epsilon_{\gamma/2}$ is a universal MUS within $\B^\epsilon_2$ for $\mathcal{F}=\mathcal{U}$ and $\mathcal{F}=\mathcal{U}_+$, we must show that all R{\'e}nyi entropies of \mbox{$\v{p}^A(\rho^\epsilon_{\gamma/2})\otimes\v{p}^B(\rho^\epsilon_{\gamma/2})$} are smaller than for any other state $\rho^\epsilon\in\B_2^\epsilon$ that can generally be described by Bloch vector $\v{q}$,
\begin{equation}\small
\label{eq:Bloch_state}
\rho^\epsilon=\frac{\iden+\v{q}\cdot\v{\sigma}}{2},\quad\v{q}=|q|(\sin\theta\cos\phi,\sin\theta\sin\phi,\cos\theta),
\end{equation}\normalsize
with $|q|\leq(1-\epsilon)$. As we will show in the next subsection it is actually sufficient to restrict the comparison to states described by
Bloch vectors
\begin{equation}
\label{eq:simplified_bloch}
\v{q}_\theta=(1-\epsilon)(\sin\theta,0,\cos\theta),\quad \theta\in[0,\gamma],
\end{equation}
i.e., one may assume $|q|=(1-\epsilon)$, $\theta\in[0,\gamma]$ and $\phi=0$ in Eq.~\eqref{eq:Bloch_state} (note that this corresponds to Bloch vectors lying between $\pm\v{a}$ and $\pm\v{b}$). Thus, the existence of a universal MUS in $\B^\epsilon_2$ can be proved by showing that among the distributions \mbox{$\v{p}^{AB}:=(p^A_1,1-p^A_1)\otimes(p^B_1,1-p^B_1)$} with
\begin{equation}\small
\label{eq:probset}
p^A_1=\frac{1+(1-\epsilon)\cos\theta}{2},\quad p^B_1=\frac{1+(1-\epsilon)\cos(\gamma-\theta)}{2},
\end{equation}\normalsize
the one with $\theta=\gamma/2$ minimises all R{\'e}nyi entropies. As our main goal is to prove that universal MUS can exist for $\mathcal{F}=\mathcal{U}$ and $\mathcal{F}=\mathcal{U}_+$, we will just focus on the particular choice of $\gamma=\pi/4$. In the last subsection of this appendix we will prove that for a noise level $\epsilon=1/2$ (Bloch vector length $1/2$) the above condition holds, so that a state specified by Eq.~\eqref{eq:qubit_uMUS} is in fact the $\epsilon$-noisy universal MUS.

\subsection*{Simplifying the set of states}

The way to prove that we can restrict the comparison to states $\rho^\epsilon_\theta$ described by Bloch vectors $\v{q}_\theta$ specified by Eq.~\eqref{eq:simplified_bloch} is to show that for every state $\rho^\epsilon\in\B_2^\epsilon$ there exists an $\epsilon$-pseudo-state $\rho^\epsilon_\theta$ such that
\begin{equation}
\label{eq:majorization}
\v{p}^A(\rho^\epsilon_\theta)\otimes\v{p}^B(\rho^\epsilon_\theta)\succ\v{p}^A(\rho^\epsilon)\otimes\v{p}^B(\rho^\epsilon).
\end{equation}
Since majorization implies trumping relation, it means that for any state in $\B_2^\epsilon$ there exists a state $\rho^\epsilon_\theta$ for which all R{\'e}nyi entropies $H_\alpha$ are lower. Hence, if  $\rho^\epsilon_{\gamma/2}$ minimises all $H_\alpha$ among $\rho^\epsilon_\theta$ states, all the remaining states in $\B^\epsilon_2$ must necessarily have higher $H_\alpha$ for all $\alpha$.

Using the transitivity of majorization we will prove our claim by restricting the subset $\B_2^\epsilon$ in a few steps, each time removing states that are ``majorized'' by states in the remaining subset. First note that for any state $\rho^\epsilon$ with \mbox{$|q|<(1-\epsilon)$} there exists a state $\tau^\epsilon$ with \mbox{$|q|=(1-\epsilon)$} and the same \mbox{$(\theta,\phi)$}, such that \mbox{$\v{p}^A(\tau^\epsilon)\succ\v{p}^A(\rho^\epsilon)$} and \mbox{$\v{p}^B(\tau^\epsilon)\succ\v{p}^B(\rho^\epsilon)$}, so that \mbox{$\v{p}^A(\tau^\epsilon)\otimes \v{p}^B(\tau^\epsilon) \succ\v{p}^A(\rho^\epsilon)\otimes \v{p}^B(\rho^\epsilon)$}. Hence we can restrict the states that we need to compare $\sigma^\epsilon$ with to states having a Bloch vector length $|q|=(1-\epsilon)$. Next we note that states with fixed $|q|$ and $\theta$ have a fixed $\v{p}^A$, but $\v{p}^B$ that depends on $\phi$. Moreover, $\v{p}^B$ for $\phi=0$ or $\phi=\pi$ majorizes all other $\v{p}^B$ with different $\phi$. Hence we can restrict the considered set of states to the ones that lie in the plane spanned by $\v{a}$ and $\v{b}$. Additionally, due to the symmetry of the problem, we only need to look at \mbox{$\theta\in[0,\pi]$} and $\gamma \in (0,\pi/2)$ (we exclude $\gamma=0$ as it is a trivial case and $\gamma=\pi/2$ because Theorem~\ref{lemma:mub} holds).

Finally we need to show that for any state $\rho^\epsilon$ that lies in this plane, has \mbox{$|q|=(1-\epsilon)$} and \mbox{$\theta\in[0,\pi]$} there exists $\rho^\epsilon_\theta$ also in that plane and with the same length of the Bloch vector, but with $\theta\in[0,\gamma]$, such that Eq.~\eqref{eq:majorization} holds. It is straightforward to check that for a state with \mbox{$\theta=\gamma$} both $\v{p}^A$ and $\v{p}^B$ majorize the corresponding distributions obtained for a state described by \mbox{$\theta\in(\gamma,\pi/2]$}. Similarly the probability distribution with \mbox{$\theta=0$} majorizes the ones described by $\theta\in[\gamma+\pi/2,\pi]$. The only thing left is to show that for every state described by \mbox{$\theta\in[\pi/2,\pi/2+\gamma]$} there is a state with \mbox{$\theta\in[0,\gamma]$} such that Eq.~\eqref{eq:majorization} holds. One can achieve this by mapping $\theta$ of every state from the first set to \mbox{$\theta-\pi/2$} in the second set. This ends the proof.

\subsection*{Proving the existence of universal MUS}

As already announced, here we will prove that among distributions \mbox{$\v{p}^{AB}:=(p^A_1,1-p^A_1)\otimes(p^B_1,1-p^B_1)$}, specified by Eq.~\eqref{eq:probset} with $\gamma=\pi/4$ and $\epsilon=1/2$, the one with $\theta=\pi/8$ minimises R{\'e}nyi entropies for all $\alpha$. In order to simplify the calculations we use a slightly different parametrization. Namely, we perform a substitution $\theta\rightarrow \theta-\pi/8$ and, due to symmetry, we only consider $\theta\in[0,\pi/8]$ (hence $\theta$, instead of measuring the angle from the $z$ axis, measures the angle from the state we want to prove is a universal MUS). To shorten the notation let us also introduce
\begin{equation*}
t_{\pm}^{\pm}=\left[2\pm\cos\left(\frac{\pi}{8}\pm\theta\right)\right]^{\alpha-1},
\end{equation*}
where the subscript refers to the sign $\pm$ in front of $\theta$.

It is straightforward to show that independently of $\alpha$ the R{\'e}nyi entropy $H_\alpha$ of the distribution $\v{p}^{AB}$ has an extremum for $\theta=0$. However, we need to show that this is the only extremum and that it is actually a minimum. Once we prove the former, the latter can be easily verified by checking that $H_\alpha$ of the distribution at the extremum $\theta=0$ is smaller than at the edge of the region $\theta=\pi/8$. To prove the uniqueness of the extremum we will show that $\frac{\partial}{\partial\theta}H_{\alpha}(\v{p}^{AB})=0$ has only a single solution for $\theta=0$. Unless $\alpha=0$ or $\alpha=1$ (which will be handled separately) vanishing of this derivative is equivalent to
\begin{equation*}
\zeta:=A \sin{\theta}+B \cos{\theta}=0,
\end{equation*}
where
\begin{eqnarray*}
A&=&\left(2 \cos\frac{\pi}{8} -\cos\theta\right) t^-_-t^-_+ -\left(2 \cos\frac{\pi}{8} +\cos\theta\right) t^+_- t^+_+,\\
B&=&\left(2 \sin\frac{\pi}{8} +\sin\theta\right) t^+_-t^-_+ -\left(2 \sin\frac{\pi}{8} -\sin\theta\right) t^-_- t^+_+.
\end{eqnarray*}
The proof consists of two main parts. First we will show that for all $\alpha\geq 2$ and $\theta\in(0,\pi/8]$ we have $\zeta<0$. We achieve this by finding a function $\zeta '$ that upper-bounds $\zeta$ in the considered parameter region and proving that it is negative. Next, we will prove that for all $\alpha\leq -3$ and $\theta\in(0,\pi/8]$ we have $\zeta>0$, this time by finding a function $\zeta '$ that lower-bounds $\zeta$ and showing that it is always positive. Finally, in the remaining region $\alpha\in[-3,2]$ the non-vanishing of $\zeta$ can be easily verified numerically (with $\alpha=0$ and $\alpha=1$ considered separately).

Let us start with $\alpha\geq 2$. For $\theta\in(0,\pi/8]$ we have that a function $\zeta '$ obtained by exchanging $A$ with \mbox{$-2 t^+_-t^+_+ \cos\theta $} upper-bounds $\zeta$ (this is because $t^+_+ t^+_-\geq t^-_- t^-_+$). To show that $\zeta '$ is always negative we divide it by the positive quantity \mbox{$t^+_-t^+_+ \cos\theta$}, and show that the obtained expression $\zeta_1+\zeta_2$ is always negative, where
\begin{eqnarray*}
\zeta_1&=&-\sin\theta+s_+\sin\theta,\\
\zeta_2&=&-\sin\theta+2s_-\sin\frac{\pi}{8},
\end{eqnarray*}
and
\begin{equation*}
s_{\pm}=\frac{t^-_+}{t^+_+}\pm \frac{t^-_-}{t^+_-}.
\end{equation*}
Now, using the fact that $(a+b)^x\geq a^x+b^x$ for $a,b>0$ and $x \geq 1$ one can easily show that $s_+<1$, which results in $\zeta_1<0$. In order to show that also $\zeta_2$ is negative it is sufficient to prove that $\zeta_2$ is a monotonically decreasing function with $\theta$ (since for $\theta=0$ it vanishes). This can be shown by upper-bounding terms dependent on \mbox{$\theta\in(0,\pi/8]$} in the expression for $\frac{\partial\zeta_2}{\partial\theta}$ that leads to
\begin{equation*}
\frac{\partial\zeta_2}{\partial\theta}\leq  -\cos\frac{\pi}{8} + \frac{8\sqrt{2}\sin\frac{\pi}{8}}{(\cos\frac{\pi}{8}-2)^4}\frac{\alpha-1}{2^\alpha}.
\end{equation*}
The above expression is maximised for \mbox{$\alpha=(1+\ln 2)/\ln 2$} and it is then negative, so that $\zeta_2$ is negative for $\theta\in(0,\pi/8]$. This ends the first part of the proof.

We now turn to the case when $\alpha\leq -3$. For \mbox{$\theta\in(0,\pi/8]$} we then have that a function $\zeta '$ obtained by exchanging the second term in the expression for $A$ by $-\left(2 \cos\frac{\pi}{8} +\cos\theta\right) t^+_-t^-_+$ lower-bounds $\zeta$ (this is because $t^+_+ \leq  t^-_+$). Moreover, substituting 1 for all $\cos\theta$ will also lower-bound the expression for $\zeta$, as for $\alpha\leq-3$ we have $B\leq 0$. We further lower-bound the expression by performing a sequence of divisions and subtractions of positive numbers: first dividing by $t^-_-t^-_+$, then subtracting $t_+^+ \sin\theta/t^-_+$ and finally dividing again by $2\sin\pi/8$. This leaves us with the lower-bound of the form:
\begin{equation*}
\sin\theta \frac{2\cos\frac{\pi}{8}\left(1-\frac{t^+_-}{t^-_-}\right)-1}{2\sin\frac{\pi}{8}}+\left(\frac{t^+_-}{t^-_-}-\frac{t^+_+}{t^-_+}\right).
\end{equation*}
As the term standing by $\sin\theta$ is a monotonically decreasing function of $\alpha$ it achieves minimum at the edge of the considered parameter space, i.e., for $\alpha=-3$. It is then straighforward to verify that it is always bigger than 1, so we can actually lower-bound $\zeta$ with
\begin{equation*}
\sin\theta+\left(\frac{t^+_-}{t^-_-}-\frac{t^+_+}{t^-_+}\right).
\end{equation*}
To show that the above equation is always positive for $\alpha\leq -3$ one can equivalently show that $\sin\theta-s_-$ is positive for $\alpha\geq 3$. This can be achieved using a method analogous to the one used to show that $\zeta_2<0$. Thus, the function lower-bounding $\zeta$ for $\alpha\leq-3$ is always positive unless $\theta=0$ and so is $\zeta$ itself. This completes the second part of the proof.

\bibliography{Bibliography_uncertainty}

\end{document}